\documentclass[lettersize,journal]{IEEEtran}
\usepackage{array}
\usepackage[caption=false,font=normalsize,labelfont=sf,textfont=sf]{subfig}
\usepackage{stfloats}
\usepackage{verbatim}
\usepackage{enumitem}
\usepackage{cite}
\usepackage{bm}
\usepackage{amsmath,amssymb,amsfonts}
\usepackage{algorithm}
\usepackage{algorithmic}
\usepackage{graphicx}
\usepackage{textcomp}
\usepackage{xcolor}
\usepackage{gensymb}
\usepackage{comment}
\usepackage{soul}
\usepackage{color, xcolor}
\usepackage{amsthm}
\usepackage{bm}
\usepackage{multicol}
\usepackage{enumitem}
\usepackage{booktabs}
\usepackage{array}
\usepackage{amsthm}

\theoremstyle{plain}

\newtheorem{lemma}{Lemma}
\newtheorem{theorem}{Theorem}

\newtheorem{remark}{Remark}
\newtheorem{result}{Result}
\newcolumntype{C}[1]{>{\centering\arraybackslash}p{#1}}
\newcolumntype{M}[1]{>{\centering\arraybackslash}m{#1}}

\def\BibTeX{{\rm B\kern-.05em{\sc i\kern-.025em b}\kern-.08em
		T\kern-.1667em\lower.7ex\hbox{E}\kern-.125emX}}
\usepackage{balance}

\begin{document}
	\title{Computational and Effective Degrees of Freedom for Spatially Stationary HMIMO Channel Modeling}
	
	\author{
		Hangsong Yan, \textit{Member}, \textit{IEEE}, Hong Yang, \textit{Senior Member}, \textit{IEEE}, Shu Sun, \textit{Senior Member}, \textit{IEEE}
		\thanks{
			Hangsong Yan is with the Hangzhou Institute of Technology, Xidian University, Hangzhou, China (email: yanhangsong@xidian.edu.cn). 
			
			Hong Yang (retired) was with the Department of Mathematics and Algorithms Research, Nokia Bell Labs, Murray Hill, USA (email: hyang.bell.labs@gmail.com). 
			
			Shu Sun is with the School of Information Science and Electronic Engineering, Shanghai Jiao Tong University, Shanghai, China (email: shusun@sjtu.edu.cn). 
	}
}
	
	\maketitle
	
\begin{abstract}
	This paper establishes a comprehensive theoretical framework for the continuous-to-discrete modeling of spatially stationary holographic MIMO channels utilizing the Nystr\"{o}m method with Gauss-Legendre quadrature (NGLQ). Starting with an operator-theoretic analysis of the NGLQ method, we prove that its quadrature error exhibits a super-exponential decay. Furthermore, we derive a spatial sampling threshold, termed computational degrees of freedom (cDoF), which reveals a $\pi/2$ oversampling penalty over the physical DoF for one-dimensional arrays, compounding to a $68\%$ computational redundancy for two-dimensional (2D) separable grids. 
	To address the ill-conditioning of the eigenvalue decomposition (EVD) problem inherent to the Nystr\"{o}m discretization, we invoke the multidimensional Szeg\H{o}-Widom asymptotic expansion. 
	This analysis yields a physically grounded semi-analytical effective DoF (eDoF) reference for 2D rectangular apertures, capturing the anisotropic boundary truncation effects to guide partial EVD and reduce computational complexity. Numerical evaluations confirm the tightness of the cDoF threshold under worst-case end-fire conditions. Moreover, simulations utilizing the closed-form kernel of the concentration operator verify that the derived eDoF reference acts as an accurate asymptotic approximation. Finally, by directly evaluating the spatial correlation kernel through high-order angular quadrature, we validate the spectral convergence behavior of the NGLQ method under non-isotropic scattering environments.
\end{abstract}

\begin{IEEEkeywords}
	Holographic MIMO, channel modeling, Nystr\"{o}m method, Gauss-Legendre quadrature, computational DoF, Szeg\H{o}-Widom asymptotic expansion, effective DoF.
\end{IEEEkeywords}

\section{Introduction}
\label{sec:Introduction}
The rapid evolution of wireless communications has spurred substantial interest in continuous aperture antennas, leading to the conceptualization of holographic multiple-input and multiple-output (HMIMO) systems. By integrating a virtually infinite number of antenna elements with infinitesimal spacing into a compact spatial aperture, HMIMO offers exceptional spatial resolution and substantial beamforming gains, pushing such systems to asymptotically approach Shannon's capacity limit~\cite{Wei2026Electromagnetic, Zhang2026Wavenumber, Bjornson2024Towards}. Motivated by these potential benefits, continuous or quasi-continuous aperture architectures have been investigated as large intelligent surfaces~\cite{Hu2018Beyond, Dardari2020Commu} and HMIMO surfaces~\cite{Gong2024HMIMO}. Practical examples include Kymeta's metasurface antennas, whose densely distributed radiating elements motivate a continuous-aperture approximation despite their discrete control~\cite{Liu2025CAPA}.
Unlike conventional discrete arrays, the behavior of HMIMO channels is characterized by a physics-based, continuous model derived from electromagnetic (EM) wave propagation theory~\cite{Pizzo2020Spatially}. 
To unlock the potential advantages of HMIMO, developing a physically compliant and mathematically tractable continuous-to-discrete channel model is of paramount importance. 

Such modeling is fundamentally a sampling and finite-dimensional representation problem for continuous fields observed over finite apertures with bounded wavenumber support.
Classical EM theory has characterized the intrinsic dimensionality of radiated fields at a prescribed accuracy and examined sampling redundancy relative to this dimension \cite{Bucci2025Degrees}. For spatially bandlimited fields, Nyquist sampling and spatial degrees of freedom (DoF) have been linked to the wavenumber-domain support~\cite{Pizzo2022Nyquist}. Reconstruction from samples confined to finite observation regions exhibits truncation errors that can be reduced through oversampling. For dense discrete planar HMIMO arrays, a spatial correlation model accounting for non-isotropic scattering and antenna directivity has been developed, which focuses on channel statistics at prescribed antenna locations~\cite{Demir2022Channel}. At the system level, the mutual information of discrete MIMO has been shown to converge to its continuous-aperture counterpart as the antenna density increases in the absence of mutual coupling~\cite{Wan2023Can}. For continuous-to-discrete HMIMO channel representation, Fourier series expansion (FSE) approximates the continuous channel using a finite set of spatial Fourier modes, with the approximation becoming asymptotically exact as the aperture dimensions grow relative to the wavelength \cite{Pizzo2022Fourier}.

In our prior work~\cite{Yan2026Spectrally}, we introduced a spectral-order Nystr\"{o}m method via Gauss-Legendre quadrature (NGLQ) for the discretization of continuous HMIMO channels. We demonstrated that the NGLQ method preserves the continuous spatial correlation properties and significantly outperforms the traditional Fourier series expansion (FSE) approach~\cite{Pizzo2022Fourier}, establishing a spectrally convergent discretization method. However, the method proposed in~\cite{Yan2026Spectrally} leaves two fundamental theoretical questions open. First, it lacks a rigorous convergence analysis; the selection of quadrature nodes in~\cite{Yan2026Spectrally} relied on the classic physical degrees of freedom (DoF) with a heuristic refinement parameter, lacking theoretical guidance. Second, to guarantee a numerically stable discrete channel model, a theoretically grounded eigenspace truncation criterion is also required, which was previously unaddressed. This requirement stems from a mathematical challenge inherent to the Nystr\"{o}m discretization: the ill-conditioning of the Karhunen-Lo\`{e}ve decomposition (KLD) eigenvalue problem when resolving the complete numerical eigenspace. Specifically, since the eigenvalues of $C^{\infty}$ smooth spatial correlation functions decay rapidly, the associated continuous kernel is effectively restricted to a finite numerical rank under a specified relative eigenvalue threshold.
We term this finite numerical rank the effective DoF (eDoF), since it quantifies the number of eigenvalues or spatial dimensions resolvable above the specified relative threshold.
To explicitly quantify this eDoF, the critical challenge is to rigorously characterize the asymptotic distribution of these eigenvalues.
Once this distribution is established, it not only uncovers the intrinsic dimensionality of the continuous channel but also provides a theoretical reference for eigenspace truncation. By safely discarding the vanishing eigencomponents, this truncation guarantees numerical stability and can also reduce computational complexity, which is particularly beneficial for large normalized apertures.

To obtain an analytically tractable reference for this eDoF, we consider the classical concentration operator~\cite{landau1980eigenvalue, Widom1982}. Specifically, quantifying the corresponding eDoF is fundamentally an eigenvalue counting problem. From a mathematical perspective, counting the eigenvalues of this operator above a prescribed threshold is equivalent to applying a discontinuous test function (i.e., a step function) and evaluating the trace of the associated integral operator. 
For one-dimensional (1D) continuous apertures, the mathematical foundation for such asymptotic trace expansions was established by Landau and Widom in~\cite{landau1980eigenvalue}, whose proof accommodated the discontinuous test functions required for eigenvalue counting. However, extending this analytical framework to multidimensional spatial apertures with piecewise smooth boundaries (e.g., two-dimensional (2D) rectangular arrays) poses a mathematical challenge. 
The foundation for higher dimensions relies on the multidimensional Szeg\H{o}-Widom asymptotic expansion, which was first conjectured by Widom in~\cite{Widom1982}. This expansion has had a profound impact on physics, notably in the analytical formulation of the entanglement entropy of fermions~\cite{PhysRevLett.96.100503}. In pure mathematics, a proof of the underlying conjecture was first achieved by Sobolev in~\cite{sobolev2013pseudo}, but it was confined to smooth domain boundaries and smooth test functions. Sobolev expanded this framework in subsequent works to accommodate realistic physical geometries. In 2015, he proved the conjecture for piecewise smooth spatial domains, demonstrating that boundary singularities such as corners do not alter the leading-order Widom asymptotics~\cite{Sobolev2015Wiener}. Furthermore, in his 2017 work \cite{Sobolev2017Functions}, the test function constraints were relaxed to encompass continuous but non-smooth functions. 
When we shift our focus to constructing the eDoF reference for spatially stationary HMIMO channels, the associated concentration operator intrinsically involves a discontinuous symbol in operator theory. Moreover, quantifying this reference also requires accommodating multidimensional piecewise smooth domains and discontinuous test functions. To the best of the authors' knowledge, the above scenario remains an open mathematical problem.  

Building upon the NGLQ method and Szeg\H{o}-Widom asymptotic expansion, this paper establishes a comprehensive theoretical framework for continuous-to-discrete HMIMO channel modeling. This framework extends the discretization method in \cite{Yan2026Spectrally} by providing a rigorous convergence analysis and theoretical guidance for quadrature-node selection and eigenspace truncation. By resolving the ill-conditioning inherent to the spectral-order Nystr\"{o}m discretization, we provide a unified methodology applicable to arbitrary aperture scales. The major contributions of this paper are summarized as follows:
\begin{itemize}[leftmargin=*, topsep=0pt, itemsep=0pt, parsep=0pt]

\item We provide an operator-theoretic analysis of the NGLQ method in the 1D case and prove that the quadrature error exhibits a super-exponential decay. Furthermore, by shifting the analytical perspective to the local interpolation remainder and applying a ratio analysis to the error bounds, we determine the theoretical boundary required for the onset of numerical convergence, which we define as the computational DoF (cDoF). The tightness of the derived cDoF is verified under the worst-case end-fire conditions.

\item We extend the operator-theoretic analysis to the 2D rectangular case via a tensor-product approach. Based on this approach, we reveal that a $\pi/2$ oversampling penalty over the physical DoF for 1D arrays is compounded to a 68\% computational redundancy for 2D separable grids. Within this 2D framework, we also propose an environment-aware NGLQ scheme to effectively reduce the required number of GLQ nodes by adapting to specific channel conditions.

\item To address the ill-conditioning of the Nystr\"{o}m discretization, we invoke the multidimensional Szeg\H{o}-Widom asymptotic expansion to construct a physically grounded, semi-analytical eDoF reference for the discontinuous test function. This asymptotic formulation captures anisotropic boundary truncation effects, naturally providing a theoretical prior to guide partial EVD solvers and reduce computational complexity for large normalized apertures. 

\item Utilizing the closed-form concentration-operator kernel, we numerically evaluate the eigenvalue distribution under discontinuous test functions. These evaluations verify that the proposed semi-analytical eDoF reference serves as an accurate asymptotic approximation, establishing an aperture-dependent, numerically stable regime for HMIMO modeling guided by the cDoF and eDoF.

\item Under the non-isotropic von Mises-Fisher (vMF) scattering model, we directly evaluate the spatial correlation kernel through high-order angular quadrature and independently verify the quadrature accuracy against closed-form reference kernels. Numerical evaluations demonstrate the spectral convergence behavior of the NGLQ framework in non-isotropic environments.
\end{itemize}

The remainder of this paper is organized as follows. Sections \ref{sec:Continuous small-scale fading model} and III establish the continuous small-scale fading model and its NGLQ-based discretization framework. Sections \ref{sec:convergence_analysis} and \ref{sec:2D_planar_apertures} derive the cDoF thresholds and prove the super-exponential quadrature convergence for 1D linear and 2D rectangular arrays, respectively. Section \ref{sec:eDoF} invokes the Szeg\H{o}-Widom asymptotic expansion to characterize the eDoF and formulates an aperture-dependent, numerically stable regime. Section \ref{sec:validation_under_NISO} validates the proposed framework under non-isotropic scattering using angular-domain quadrature.
Finally, Section \ref{sec:conclusion} concludes the paper.

\section{Continuous Small-Scale Fading Model}
\label{sec:Continuous small-scale fading model}
This section establishes the physical model for the monochromatic, far-field, continuous small-scale fading, denoted by $h(\mathbf{r})$, in a source-free environment, where $\mathbf{r} = (x, y, z)$ represents the generic 3D spatial coordinate.
Building upon the plane-wave representation introduced in~\cite{Pizzo2020Spatially}, the small-scale fading is characterized as a scalar random field fundamentally governed by the physical laws of wave propagation.
Specifically, this field satisfies the homogeneous scalar Helmholtz equation, which constrains the Cartesian components of the wavenumber vector $(k_x, k_y, k_z)$ onto a sphere of radius $\kappa = 2\pi/\lambda$ (with $\lambda$ being the wavelength)~\cite{hildebrand1962advanced}:
\begin{equation}
	k_x^2 + k_y^2 + k_z^2 = \kappa^2.
	\label{eq:k_x_k_y_k_z}
\end{equation}

Under the far-field propagation assumption, we can neglect the effects of evanescent waves (i.e., $k_z^2 < 0$), since they decay exponentially with distance and are negligible far from the source. This assumption yields two real solutions for the vertical wavenumber component, $k_z = \pm\sqrt{\kappa^2 - k_x^2 - k_y^2}$, physically corresponding to the upgoing $(+)$ and downgoing $(-)$ propagating waves impinging from the upper and lower hemispheres in a 3D scattering environment. Consequently, the total field at any spatial position $\mathbf{r}$ is the superposition of these two components, i.e., $h(\mathbf{r}) = h_{+}(\mathbf{r}) + h_{-}(\mathbf{r})$. These components are expressed via the inverse Fourier transform over the constraint region (i.e., the disk $k_x^2 + k_y^2 \leq \kappa^2$):
\begin{align}
	&h_{\pm}(\mathbf{r}) = \frac{1}{4\pi\sqrt{\pi}}\iint_{k_x^2+k_y^2 \leq \kappa^2} \frac{A_{h,\pm}(k_x,k_y)}{(\kappa^2 - k_x^2 - k_y^2)^{1/4}} \nonumber \\
	&\quad\quad \times W^{\pm}(k_x, k_y)e^{j(k_x x + k_y y \pm \sqrt{\kappa^2 - k_x^2 -k_y^2}z)}dk_xdk_y,
	\label{eq:h_pm_x_y_z}
\end{align}
where the prefactor $\frac{1}{4\pi\sqrt{\pi}}$ ensures the unit-variance normalization of the fading field, and the denominator $(\kappa^2 - k_x^2 - k_y^2)^{1/4}$ originates from the Jacobian determinant for the spherical-to-planar wavenumber projection.
Furthermore, $A_h(\mathbf{k}) = A_h(k_x, k_y, k_z)$ denotes a real-valued, non-negative deterministic field, called the spectral factor, whose square, $A_h^2(k_x, k_y, k_z)$, describes the angular power distribution of the arriving plane waves.
The projected spectral factor $A_{h,\pm}(k_x,k_y)$ in (\ref{eq:h_pm_x_y_z}) is given by:
\begin{equation}
	A_{h, \pm}(k_x, k_y) = A_{h}(k_x, k_y, \pm\sqrt{\kappa^2 - k_x^2 - k_y^2}).
	\label{eq:A_h_pm}
\end{equation}
Lastly, $W^{\pm}$ are two independent, zero-mean, complex-valued Gaussian white-noise random fields with unit variance. 
	
\section{NGLQ Method}
\label{sec:NGLQ_method}
In this section, we present the NGLQ method to establish a tractable, discrete model for HMIMO small-scale fading. This approach relies on discretizing the continuous KLD while preserving the channel's spatial correlation properties.

\subsection{Karhunen-Lo\`eve Decomposition}
\label{sec:KLD_procedure}
For a channel with spatial autocorrelation function (i.e., kernel) $K(\mathbf{r}, \mathbf{r}') = \mathbb{E}[h(\mathbf{r})h^*(\mathbf{r}')]$, the Hermitian and positive-semidefinite nature of this kernel allows us to invoke Mercer's theorem~\cite{riesz1990functional, ghojogh2021reproducing}. This theorem guarantees that over any compact spatial domain $\mathcal{D}$, the continuous kernel admits an exact spectral series expansion:
\begin{equation}
	K(\mathbf{r}, \mathbf{r}') = \sum_{n=1}^{\infty}\mu_n g_n(\mathbf{r})g_n^*(\mathbf{r}'),
	\label{eq:Mercer}
\end{equation}
where $\{\mu_n\}$ are the real, non-negative eigenvalues and $\{g_n(\mathbf{r})\}$ constitute the corresponding complete orthonormal basis of eigenfunctions.
By projecting the kernel onto this basis and leveraging its orthonormality, the eigenfunctions are shown to satisfy the continuous Fredholm integral equation of the second kind, widely recognized in stochastic processes as the KLD problem:
\begin{equation}
	\int_{\mathcal{D}}K(\mathbf{r}, \mathbf{r}')g_n(\mathbf{r}')\,d\mathbf{r}' = \mu_n g_n(\mathbf{r}), \quad n = 1,2,\ldots
	\label{eq:KLD}
\end{equation}

\subsection{A Spectral-Order Nystr\"{o}m Method via Gauss-Legendre Quadrature}
\label{sec:NGLQ}
To numerically solve the continuous Fredholm integral equation in (\ref{eq:KLD}), we can apply GLQ over the compact region $\mathcal{D}$. For any sufficiently smooth function $g(\mathbf{r})$, the integral is approximated as:
\begin{equation}
	\int_{\mathcal{D}}g(\mathbf{r})\,d\mathbf{r} \approx \sum_{m=1}^{M}w_m g(\mathbf{r}_m)
	\label{eq:quadrature_formula}
\end{equation}
where $\{\mathbf{r}_m\}_{m=1}^M$ and $\{w_m\}_{m=1}^M$ denote the predefined GLQ nodes and their strictly positive weights, respectively. Substituting (\ref{eq:quadrature_formula}) into (\ref{eq:KLD}), we obtain the discretized eigenvalue problem:
\begin{equation}
	\tilde{\mu}_n\tilde{g}_n(\mathbf{r}_m) = \sum_{m'=1}^{M}w_{m'}K(\mathbf{r}_m, \mathbf{r}'_{m'})\tilde{g}_n(\mathbf{r}'_{m'}).
	\label{eq:Discrete_KLD}
\end{equation}
Here, $\tilde{\mu}_n$ and $\tilde{g}_n$ represent the numerical approximations of the continuous eigenvalues and eigenfunctions, respectively.
Equation (\ref{eq:Discrete_KLD}) is equivalent to the following eigenvalue decomposition (EVD) problem:
\begin{equation}
	\mathbf{KW}\tilde{\mathbf{g}}_n = \tilde{\mu}_n\tilde{\mathbf{g}}_n,
	\label{eq:KLD_Nystrom_EVD}
\end{equation}
where $\mathbf{K}_{m,m'} = K(\mathbf{r}_m,\mathbf{r}'_{m'})$, $\mathbf{W} = \text{diag}\{[w_1, \ldots, w_M]\}$, and $\tilde{\mathbf{g}}_n = [\tilde{g}_n(\mathbf{r}_1), \ldots, \tilde{g}_n(\mathbf{r}_M)]^T$. Since $\mathbf{KW}$ is generally non-Hermitian, a standard EVD may yield numerical instabilities. To resolve this, (\ref{eq:KLD_Nystrom_EVD}) can be transformed via a similarity transformation using the positive definite matrix $\mathbf{W}^{1/2}$:
\begin{equation}
	(\mathbf{W}^{1/2}\mathbf{K}\mathbf{W}^{1/2})\mathbf{W}^{1/2}\tilde{\mathbf{g}}_n = \tilde{\mu}_n\mathbf{W}^{1/2}\tilde{\mathbf{g}}_n,
	\label{eq:Nystrom_EVD_Symmetric}
\end{equation}
which reformulates the problem into the EVD of a Hermitian matrix, $\mathbf{W}^{1/2}\mathbf{K}\mathbf{W}^{1/2}$. Once the discrete eigenvectors are obtained, the $n$-th continuous eigenfunction can be recovered by substituting $\tilde{\mathbf{g}}_n$ from (\ref{eq:Nystrom_EVD_Symmetric}) into (\ref{eq:Discrete_KLD}):
\begin{equation}
	\hat{g}_n(\mathbf{r}) = \frac{1}{\tilde{\mu}_n}\sum_{m'=1}^M w_{m'}K(\mathbf{r}, \mathbf{r}'_{m'})\tilde{g}_n(\mathbf{r}'_{m'}), \; \; n = 1,\ldots, N_{\text{trunc}}.
	\label{eq:eigenfunction_interpolation}
\end{equation}
Note that the index $n$ is truncated at $N_{\text{trunc}}$ ($N_{\text{trunc}} \leq M$). This truncation is practically necessary because the eigenvalues $\tilde{\mu}_n$ of a smooth kernel decay rapidly toward zero, and a near-zero $\tilde{\mu}_n$ in the denominator of (\ref{eq:eigenfunction_interpolation}) would amplify numerical noise.

Based on the interpolated eigenfunctions and the calculated eigenvalues, the NGLQ-based HMIMO small-scale fading $\hat{h}(\mathbf{r})$ is synthesized as a truncated spatial series:
\begin{equation}
	\hat{h}(\mathbf{r}) = \sum_{n=1}^{N_{\text{trunc}}} a_n \sqrt{\tilde{\mu}_n}\hat{g}_n(\mathbf{r}),
	\label{eq:channel_model}
\end{equation}
where $a_n \sim \mathcal{CN}(0, 1)$ are independent and identically distributed (i.i.d.) complex Gaussian random variables.

\section{Computational DoF and Convergence Analysis}
\label{sec:convergence_analysis}
We first consider the 1D spatial domain to establish a convergence baseline for the proposed method. The continuous integral operator $\mathcal{K}$ acting on a function $g$ over a linear aperture of length $L$ is defined as:
\begin{equation}
	(\mathcal{K}g)(x) = \int_{0}^LK(x,x')g(x')\,dx'.
	\label{eq:1D_integral_operator}
\end{equation}
The corresponding discrete NGLQ operator is then given by:
\begin{equation}
	(\mathcal{K}_M g)(x) = \sum_{m=1}^{M}w_mK(x,x'_m)g(x'_m),
	\label{eq:1D_NGLQ_operator}
\end{equation}
where $\{x'_m\}_{m=1}^M \subset [0, L]$ denote the physical Gauss-Legendre nodes. To evaluate the quadrature truncation error, we define the integrand parameterized by the observation point $x$ as $f_x(x') = K(x,x')g(x')$. Consequently, the $M$-th order GLQ error in the physical domain is defined as:
\begin{equation}
	E_M(f_x) = \int_{0}^{L}f_x(x')\,dx' - \sum_{m=1}^{M}w_mf_x(x'_m).
	\label{eq:GLQ_error_1D}
\end{equation}

To analytically evaluate the quadrature truncation error, we map the physical coordinates $x, x' \in [0, L]$ to the standard GLQ interval $s, t \in [-1, 1]$ via the affine transformations $x = \frac{L}{2}s + \frac{L}{2}$ and $x' = \frac{L}{2}t + \frac{L}{2}$. Correspondingly, we define the mapped continuous kernel and its eigenfunction on the standard interval as $\bar{K}(s, t) = K(x(s), x'(t))$ and $\bar{g}(t) = g(x'(t))$, respectively.
Since the integration measure scales as $dx' = \frac{L}{2}dt$ and the weights map as $w_m = \frac{L}{2}\tilde{w}_m$ (where $\tilde{w}_m$ are the standard GLQ weights), the approximation error of the operator can be factored as $\frac{L}{2}E_{M}(f)$, where $E_{M}(f)$ is the baseline quadrature error on the standard interval:
\begin{equation}
	E_M(f) = \int_{-1}^{1}f(t)\,dt - \sum_{m=1}^{M}\tilde{w}_mf(t_m).
	\label{eq:GLQ_error_1D_std}
\end{equation}
Here, $f(t) = \bar{K}(s, t)\bar{g}(t)$ is defined as the standardized integrand. For notational simplicity, the parametric dependence of $f(t)$ on the observation point $s$ (or $x$) is omitted in the sequel. We now establish the following theorem for $E_M(f)$ in the context of HMIMO channel modeling.
\begin{theorem}
\label{theorem:super_exponential_threshold}
	For a spatially stationary 1D HMIMO aperture of physical length $L$ operating under a maximum spatial wavenumber $\kappa = 2\pi/\lambda$, the standardized quadrature error $E_M(f)$ under an $M$-th order GLQ approximation satisfies the following properties:
	
	1) Onset of Convergence: The sufficient computational threshold at which $E_M(f)$ initiates its super-exponential decay is
	\begin{equation}
		M_c \approx \frac{\pi L}{\lambda}.
		\label{eq:M_threshold}
	\end{equation}
	
	2) Global Upper Bound: $E_M(f)$ exhibits super-exponential convergence and is strictly upper bounded by:
	\begin{equation}
		|E_M(f)| \leq \tilde{C}(M)\Big(\frac{e\pi L}{2\lambda M}\Big)^{2M}
		\label{eq:envelope_E_M_f}
	\end{equation}
	where $\tilde{C}(M) = \frac{\sqrt{\pi M}}{2M+1}\exp\left(\frac{1}{3M} - \frac{3}{24M+1}\right)\sup_{t \in \mathbb{R}}|f(t)|$ is an algebraically decaying pre-factor, and $e$ is Euler's number.
\end{theorem}
The proof for Theorem \ref{theorem:super_exponential_threshold} is provided in Appendix~\ref{appendix:super_exponential_threshold}.
\begin{remark}
It is important to note the theoretical distinction between the onset of convergence—derived via local interpolation analysis—and the global upper bound—established through quadrature error bounds. As corroborated by the simulation results in Fig.~\ref{Fig:quadrature_error_cos}, the actual quadrature error exhibits a sharp decay at $M \approx \pi L/\lambda$. This onset threshold is captured by the sequence ratio analysis of the interpolation remainder bound, which establishes the spatial sampling requirement necessary to resolve the underlying oscillatory kernels. On the other hand, the closed-form bound in (\ref{eq:envelope_E_M_f}) provides a global upper bound across the entire bandlimited function space. By depending on the $2M$-th derivative, this bound incorporates the $2M-1$ algebraic precision inherent to the NGLQ method. Thus, while the sequence ratio analysis identifies the convergence threshold ($\pi L/\lambda$), the quadrature bound determines the $(C/M)^{2M}$ super-exponential decay rate after this computational threshold is exceeded.
\end{remark}

We define the derived threshold $M_c \approx \pi L/\lambda$ as the cDoF of the continuous 1D HMIMO channel. Notably, this threshold exceeds the fundamental physical DoF ($2L/\lambda$). 
The ratio cDoF/DoF = $\pi/2$ reveals a fundamental requirement of the proposed NGLQ discretization: to numerically evaluate the continuous HMIMO channel, the discrete computational grid demands an oversampling penalty by a factor of $\pi/2$ to guarantee the onset of super-exponential quadrature convergence.

\subsection{Tightness of the cDoF Threshold}
\label{sec:Tightness of the cDoF Threshold}
While we have established the sufficient condition of $M \approx \frac{\pi L}{\lambda}$ for the onset of super-exponential quadrature convergence, a natural question arises regarding the tightness of this theoretical threshold.
In this subsection, we demonstrate that our derived cDoF is, in fact, tight. It serves as a necessary safeguard against the worst-case spatial variation, represented by the dual end-fire channel.

\begin{figure}[!t]
	\centerline{\includegraphics[scale=0.14]{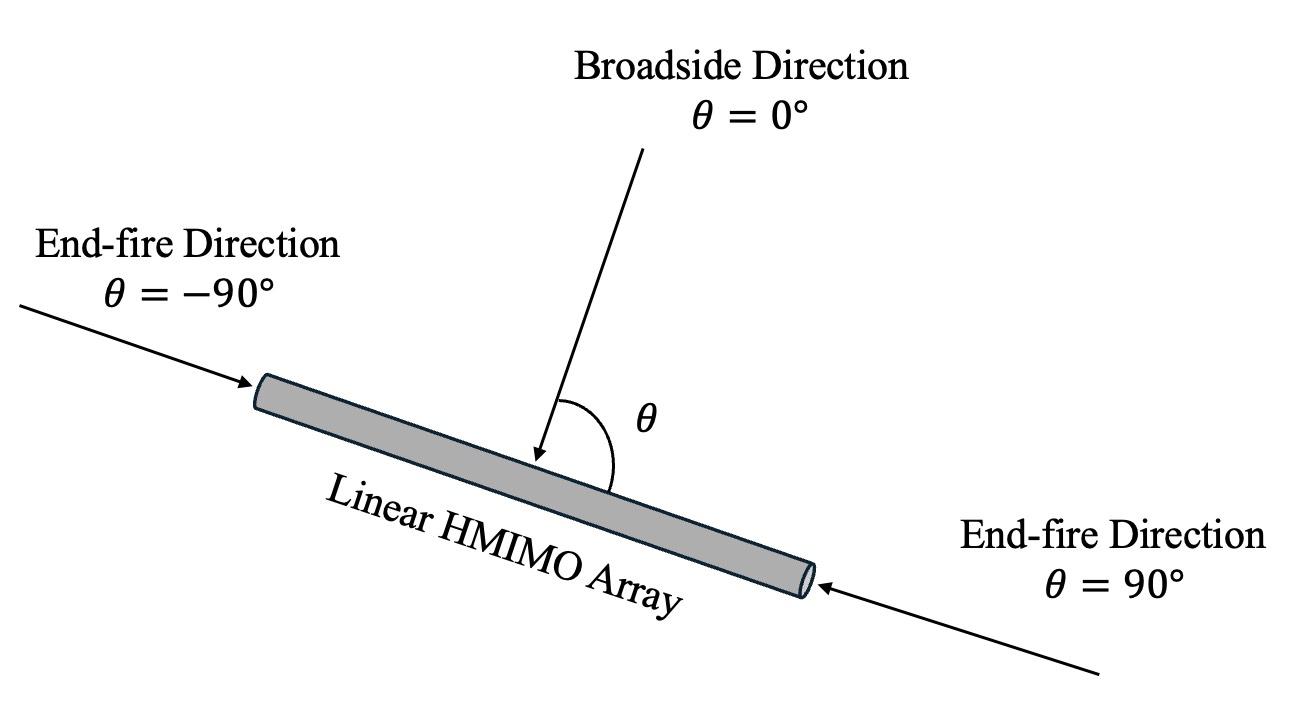}}
	\caption{Illustration of end-fire and broadside incidence scenarios.}
	\label{Fig:end_fire_broadside}
\end{figure}
Consider a deterministic, scattering-free environment where the linear HMIMO array is illuminated exclusively by two plane waves impinging from the exact left and right end-fire directions (i.e., $\theta = \pm 90^{\circ}$, as illustrated in Fig.~\ref{Fig:end_fire_broadside}). Physically, the projection of the wave vector in these directions generates the maximal spatial phase variations, $\pm 2\pi/\lambda$. Therefore, the PSD is composed of two Dirac delta impulses:
\begin{equation}
	S(k) = \frac{1}{2}\delta\Big(k - \frac{2\pi}{\lambda}\Big) + \frac{1}{2}\delta\Big(k+\frac{2\pi}{\lambda}\Big).
	\label{eq:S_k_end_fire}
\end{equation}
The mapped kernel function $\bar{K}(s, t)$ degenerates into a simple and separable harmonic form:
\begin{equation}
	\bar{K}(s, t) = \cos(\omega_{\max}s)\cos(\omega_{\max}t) + \sin(\omega_{\max}s)\sin(\omega_{\max}t),
	\label{eq:kernel_end_fire}
\end{equation}
where $\omega_{\max} = \pi L/\lambda$. Mathematically, the rank of the integral operator based on this separable kernel collapses to exactly $2$. Using the KLD integral equation in (\ref{eq:KLD}), it is straightforward to verify that the two eigenfunctions are: 
\begin{equation}
	\bar{g}_1(t) = \cos(\omega_{\max} t), \qquad  \bar{g}_2(t) = \sin(\omega_{\max} t),
\end{equation}
both of which consist of a single frequency component $\omega_{\max}$.

Since the integrand is $f_n(t) = \bar{K}(s,t)\bar{g}_n(t)$, we consider $f_1(t)$ as an example. By applying basic trigonometric identities, $f_1(t)$ is expressed as:
\begin{align}
	&f_1(t) = \nonumber \\
	&\cos(\omega_{\max} s)\cdot\frac{1 + \cos(2\omega_{\max}t)}{2} + \sin(\omega_{\max} s)\cdot\frac{\sin(2\omega_{\max} t)}{2}.
	\label{eq:f_1_t}
\end{align}
Equation (\ref{eq:f_1_t}) reveals that $f_1(t)$ exhibits a maximum spatial frequency component of $2\omega_{\max} = 2\pi L/\lambda$, which aligns with the exponential type $2\pi L/\lambda$ established in Lemma~\ref{lemma:frequency_doubling} of Appendix~\ref{appendix:super_exponential_threshold}.

\begin{figure}[!t]
	\centerline{\includegraphics[scale=0.39]{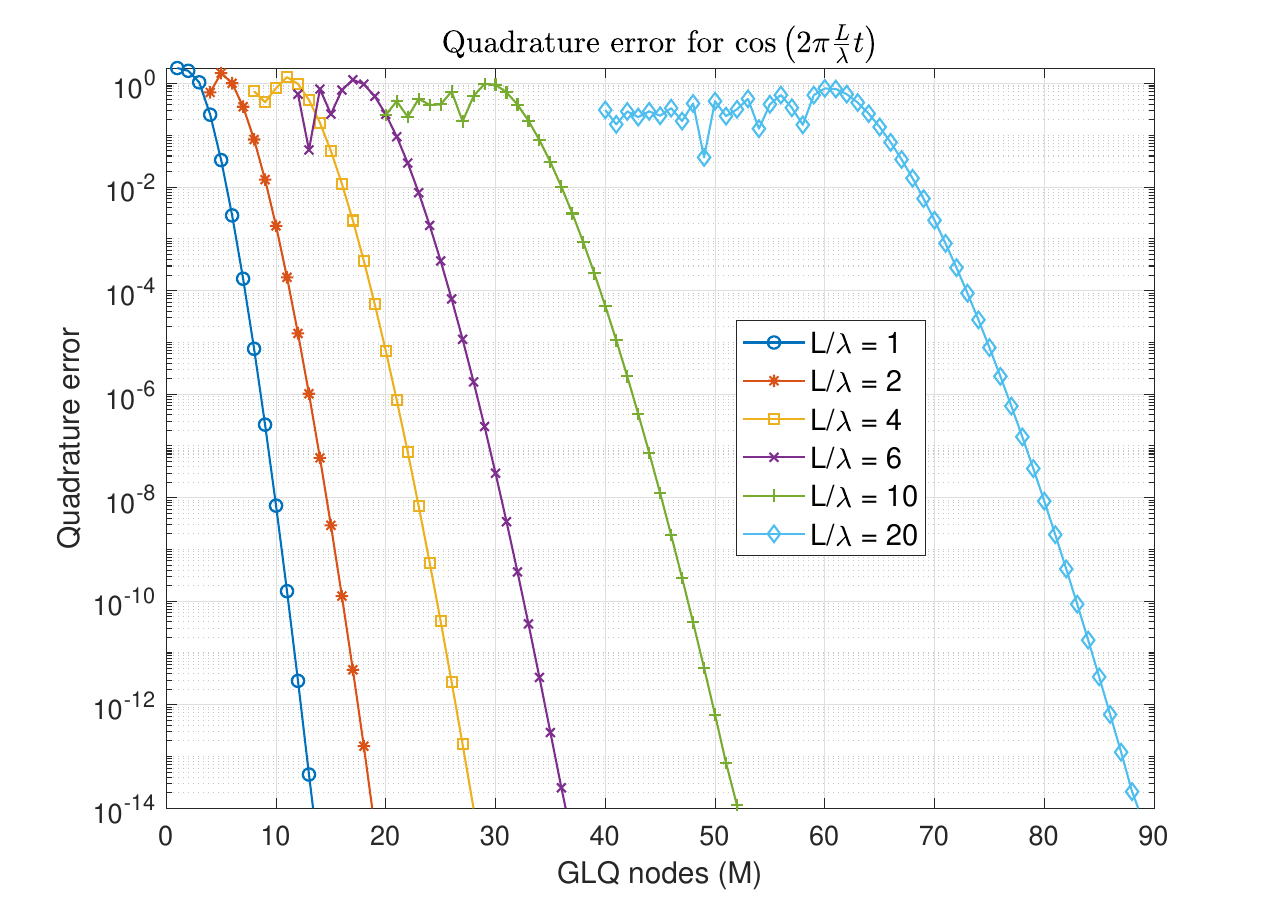}}
	\caption{Illustration of quadrature errors for $\cos(2\pi \frac{L}{\lambda}t)$ with respect to the number of GLQ nodes.}
	\label{Fig:quadrature_error_cos}
\end{figure}
To numerically validate this, Fig.~\ref{Fig:quadrature_error_cos} plots the absolute quadrature error of $\cos(2\pi\frac{L}{\lambda}t)$ against the number of GLQ nodes $M$. For all considered normalized aperture sizes $\frac{L}{\lambda}$, the error exhibits a super-exponential decay upon entering the convergence region. Furthermore, the numerical convergence thresholds align well with the theoretical bound $M \approx \pi L/\lambda$ derived in (\ref{eq:M_threshold}).

\subsection{Verification through PSWFs}
To further validate the proposed framework, we evaluate the NGLQ method using the classic sinc kernel. We consider the following integral equation:
\begin{equation}
	\int_{-1}^{1}\frac{\sin(c(s-t))}{\pi(s-t)}\phi_n(t)\,dt = \mu_n\phi_n(s),
	\label{eq:1D_PSWF}
\end{equation}
where $c = \frac{\pi L}{\lambda}$ is the spatial bandwidth parameter and $\phi_n(t)$ denotes the $n$-th classic prolate spheroidal wave function (PSWF)~\cite{Slepian1961PSWF}. The NGLQ method is applied to numerically discretize the kernel $\frac{\sin(c(s-t))}{\pi(s-t)}$ and reconstruct the corresponding eigenfunctions $\phi_n(t)$.
\begin{figure}[!t]
	\centerline{\includegraphics[scale=0.35]{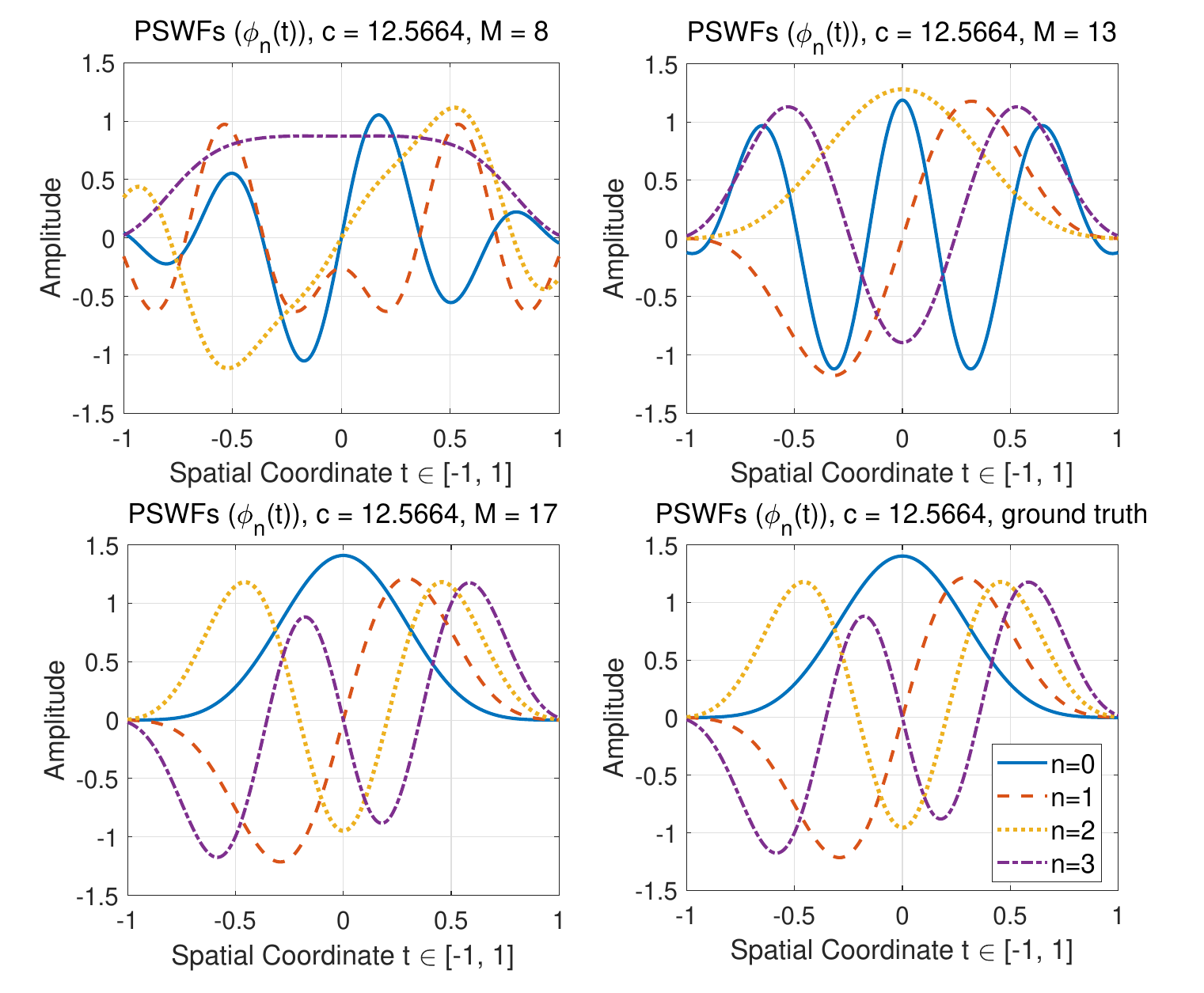}}
	\caption{PSWF reconstruction against the number of GLQ nodes, $M$, with $c = \pi L/\lambda = 4\pi$.}
	\label{Fig:PSWF_Reconstruction}
\end{figure}
Fig.~\ref{Fig:PSWF_Reconstruction} illustrates the reconstructed PSWFs for $c = 4\pi$. The ground-truth PSWFs are computed by applying the Legendre-Galerkin spectral method to the commuting Sturm-Liouville differential operator~\cite{Slepian1961PSWF}. For this specific configuration, the physical DoF is $2L/\lambda = 8$, while the derived cDoF is $\pi L/\lambda \approx 12.57$. As the number of GLQ nodes $M$ surpasses the convergence threshold (i.e., $M = 13$), the waveforms converge rapidly. Indeed, the reconstructed PSWFs obtained at $M = 17$ are virtually indistinguishable from the ground truth.

\section{Extension to 2D Planar Apertures: Tensor-Product Convergence Analysis}
\label{sec:2D_planar_apertures}
While the 1D analysis establishes the fundamental connection between quadrature nodes and physical DoF, practical HMIMO systems typically employ two-dimensional (2D) planar apertures. In this section, we extend our analytical framework to a 2D rectangular aperture of size $L_x \times L_y$.

\subsection{2D Integral Operator and Tensor-Product NGLQ}
Consider a spatially stationary 2D continuous HMIMO aperture. By applying the standard affine transformations to the physical $x$- and $y$-dimensions, the spatial coordinates are mapped to the standard 2D domain $(u, v) \in [-1, 1] \times [-1, 1]$. The 2D normalized integrand is defined as $f(u, v) = \bar{K}(u', v', u, v)\bar{g}(u, v)$, where $(u', v')$ denotes the normalized observation point.

Let $\mathcal{I}_w$ denote the 1D exact continuous integration operator with respect to $w$ over $[-1, 1]$, and let $\mathcal{Q}_w^{(M)}$ denote its corresponding NGLQ operator using $M$ numerical nodes. The exact 2D integration and its tensor-product NGLQ approximation are straightforwardly formulated as:
\begin{align}
	\mathcal{I}_u\mathcal{I}_v f &= \int_{-1}^{1}\int_{-1}^{1}f(u, v)\,dv\,du, \\
	\mathcal{Q}_u^{(M_x)}\mathcal{Q}_v^{(M_y)} f &= \sum_{m=1}^{M_x}\sum_{n=1}^{M_y} w_m^{(u)}w_n^{(v)} f(u_m, v_n).
\end{align}
Consequently, the 2D quadrature error $E_{M_x, M_y}(f)$ is defined by the operator difference:
\begin{equation}
	E_{M_x, M_y}(f) = \left(\mathcal{I}_u\mathcal{I}_v - \mathcal{Q}_u^{(M_x)}\mathcal{Q}_v^{(M_y)}\right) f.
	\label{eq:2D_error_def}
\end{equation}

\subsection{2D Error Decomposition and Multivariate Properties}
To analyze $E_{M_x, M_y}(f)$ without relying on intricate multivariate error analysis, we decompose the 2D operator error into a linear combination of 1D errors. By adding and subtracting the mixed operator term $\mathcal{Q}_u^{(M_x)}\mathcal{I}_v f$, (\ref{eq:2D_error_def}) can be decoupled as:
\begin{align}
	& E_{M_x, M_y}(f) \nonumber \\ & \quad = (\mathcal{I}_u\mathcal{I}_v f - \mathcal{Q}_u^{(M_x)}\mathcal{I}_v f)  + (\mathcal{Q}_u^{(M_x)}\mathcal{I}_v f - \mathcal{Q}_u^{(M_x)}\mathcal{Q}_v^{(M_y)} f) \notag \\
	& \quad = E_{M_x}\left(\mathcal{I}_v f\right) + \mathcal{Q}_u^{(M_x)}\left(E_{M_y}\big(f(u, \cdot)\big)\right),
	\label{eq:2D_error_decomp}
\end{align}
where $E_{M_x}(\cdot)$ and $E_{M_y}(\cdot)$ are the 1D quadrature errors acting along their respective axes. By the triangle inequality, the absolute 2D quadrature error is bounded by:
\begin{align}
	\left|E_{M_x, M_y}(f)\right| \leq & \left|E_{M_x}\left(\mathcal{I}_v f\right)\right| + \sum_{m=1}^{M_x}w_m^{(u)}\left|E_{M_y}\left(f(u_m, \cdot)\right)\right|.\label{eq:2D_error_triangle}
\end{align}
To bound the components in (\ref{eq:2D_error_triangle}), we must establish the analytical properties of $f(u,v)$ induced by the constraints of the 2D wavenumber domain. The following lemma extends the 1D Paley-Wiener theorem to the multivariate case.
\begin{lemma}[Multivariate Extension of Paley-Wiener Theorem]
	For a 2D spatially stationary continuous HMIMO aperture operating under the spatial wavenumber $\kappa = 2\pi/\lambda$, the 2D integrand $f(u, v)$ is a multivariate entire function. Specifically, it is of exponential type $2\pi L_x/\lambda$ with respect to $u$, and of exponential type $2\pi L_y/\lambda$ with respect to $v$.
\end{lemma}
\begin{proof}
	The physical constraints of the 2D wavenumber domain restrict the wave vectors $(k_x, k_y)$ to the spectral disk $k_x^2 + k_y^2 \leq \kappa^2$. Since this circular physical spectral disk is enclosed by the rectangular bounding domain $[-\kappa, \kappa] \times [-\kappa, \kappa]$, the maximum spatial frequencies along the orthogonal axes are guaranteed to satisfy $|k_x| \leq \kappa$ and $|k_y| \leq \kappa$.
	
	According to the multivariate Paley-Wiener theorem for rectangular supports~\cite[Chapter 3]{nikolskii1975approximation}, restricting the Fourier support to a separable rectangular region $[-a, a] \times [-b, b]$ directly ensures that the corresponding entire function possesses separable exponential types $a$ and $b$ along its respective axes.
	Therefore, by incorporating the coordinate scaling factors ($L_x/2$ and $L_y/2$) from the standard affine transformations, and following the frequency-doubling effect inherent to the integrand product established in Lemma \ref{lemma:frequency_doubling} of Appendix~\ref{appendix:super_exponential_threshold}, the integrand $f(u, v)$ structurally inherits the decoupled exponential types of $2\pi L_x/\lambda$ for $u$ and $2\pi L_y/\lambda$ for $v$. 
\end{proof}
Exploiting this separable exponential type property, we can now apply Bernstein's inequality~\cite{Rahman2009Bernstein} for entire functions to the partial derivatives along each independent axis. Since the remainder of an $M$-point GLQ rule depends on the $2M$-th derivative, the bounds parameterized by the GLQ orders $M_x$ and $M_y$ are given by:
\begin{align}
	\max_{(u,v)\in[-1,1]^2}\left|\frac{\partial^{2M_x}f(u, v)}{\partial u^{2M_x}}\right|
	\leq \left(\frac{2\pi L_x}{\lambda}\right)^{2M_x} \sup_{(u,v)\in\mathbb{R}^2} |f(u, v)|,
	\label{eq:partial_u_bound}
\end{align}
\begin{align}
	\max_{(u,v)\in[-1,1]^2}\left|\frac{\partial^{2M_y}f(u, v)}{\partial v^{2M_y}}\right| 
	\leq \left(\frac{2\pi L_y}{\lambda}\right)^{2M_y} \sup_{(u,v)\in\mathbb{R}^2} |f(u, v)|.
	\label{eq:partial_v_bound}
\end{align}

\subsection{Super-Exponential Convergence of 2D HMIMO Arrays}
Building upon the operator decomposition and the partial derivative bounds established above, we now present the global convergence theorem for 2D rectangular apertures.
\begin{theorem}\label{theorem:2D_convergence}
	For a 2D rectangular HMIMO aperture of size $L_x \times L_y$, the absolute quadrature error of reconstructing the continuous channel via a tensor-product NGLQ method is strictly bounded by the sum of two independent super-exponential upper bounds:
	\begin{align}
		&|E_{M_x, M_y}(f)| \nonumber \\
		 & \quad \leq \tilde{C}(M_x)\left(\frac{e\pi L_x}{2\lambda M_x}\right)^{2M_x} + \tilde{C}(M_y)\left(\frac{e\pi L_y}{2\lambda M_y}\right)^{2M_y},
		\label{eq:2D_envelope}
	\end{align}
	where $\tilde{C}(M_x)$ and $\tilde{C}(M_y)$ are algebraically decaying pre-factors proportional to the global supremum of $|f(u,v)|$ over $\mathbb{R}^2$. Furthermore, a sufficient condition for the computational onset of 2D numerical convergence is that the numbers of discrete nodes satisfy:
	\begin{equation}
		M_x > \frac{\pi L_x}{\lambda} \quad \text{and} \quad M_y > \frac{\pi L_y}{\lambda}.
		\label{eq:2D_DoF}
	\end{equation}
\end{theorem}
The detailed derivation is provided in Appendix~\ref{appendix:2D_convergence}.
\begin{remark}
	It is important to observe the discrepancy between the convergence threshold of the tensor-product NGLQ (requiring a cDoF of $M_x M_y \approx \pi^2A_{\text{array}}/\lambda^2$ where $A_{\text{array}} = L_xL_y$) and the classic physical DoF for a 2D planar aperture ($\pi A_{\text{array}}/\lambda^2$). The proposed grid yields an efficiency ratio (physical DoF to cDoF) of $1/\pi \approx 32\%$.
	
	This $68\%$ computational redundancy is a twofold penalty. First, the 1D quadrature oversampling penalty (i.e., the factor of $\pi/2$) is geometrically squared in 2D to $\pi^2/4$. Second, the separability penalty introduces a geometric mismatch: while physical propagating waves are confined to a spectral disk ($k_x^2 + k_y^2 \leq \kappa^2$), the separable tensor-product operator forces the GLQ nodes to resolve a larger rectangular bounding box ($|k_x|\leq \kappa, |k_y|\leq \kappa$). Consequently, redundant computational nodes are inevitably allocated to the non-propagating spectral corners ($k_x^2 + k_y^2 > \kappa^2$) where the actual propagating energy is zero. This observation underscores a profound trade-off: while separable grids offer mathematical tractability, achieving the optimal 2D continuous-to-discrete mapping may require non-separable spatial sampling strategies in future HMIMO architectures.
\end{remark}

\subsection{Environment-Aware NGLQ}
Given the analysis in Section~\ref{sec:Tightness of the cDoF Threshold}, it is evident that the required number of GLQ nodes is related to the maximal incidence angle, or equivalently, the maximum projected spatial wavenumber. In practical wireless network deployments (e.g., cellular network sectors), the angular spectrum is often restricted within a specific angular sector, especially in terms of vertical coverage. In this subsection, we demonstrate how the NGLQ framework intrinsically adapts to such environment-aware constraints, leading to a reduction in the required cDoF.

We consider a 2D rectangular aperture located in the $xy$-plane, with the $x$- and $y$-axes representing its horizontal and vertical dimensions, respectively, and the $z$-axis along its broadside direction. We note that for a monochromatic plane wave, the unit vector $\mathbf{u}$ specifying its propagation direction in physical space is related to the corresponding wavenumber vector $\mathbf{k}$ as $\mathbf{k} = \kappa \mathbf{u}$ with $\kappa = 2\pi/\lambda$~\cite{Pizzo2022Fourier}. Therefore, angular restrictions on the propagation directions translate directly into restrictions on the wavenumber support.

To describe the vertical and horizontal coverage explicitly, we introduce the angles $\alpha$ and $\beta$. Specifically, $\alpha$ is the signed angle between $\mathbf{u}$ and the horizontal $xz$-plane, positive toward the $+y$-axis. The horizontal angle $\beta$ is the signed angle from the $+z$-axis to the projection of $\mathbf{u}$ onto the $xz$-plane, positive toward the $+x$-axis. These angles directly specify the vertical and horizontal coverage relative to the aperture orientation. Under this parameterization, we have 
\begin{equation}
	\mathbf{k}(\alpha, \beta) = \kappa[\cos\alpha\sin\beta, \sin\alpha, \cos\alpha\cos\beta]^T.
	\label{eq:k_alpha_beta}
\end{equation}
Since the aperture lies in the $xy$-plane, $k_x = \kappa\cos\alpha\sin\beta, k_y=\kappa\sin\alpha$ govern the spatial phase variation across the aperture.
We assume that the angular range is restricted to $\alpha \in [-\alpha_{\max}, \alpha_{\max}]$ and $\beta\in[-\beta_{\max}, \beta_{\max}]$, where $\alpha_{\max}, \beta_{\max}\in[0, \pi/2]$. The corresponding support in the Cartesian wavenumber domain $(k_x, k_y)$ is restricted accordingly.
To construct the separable tensor-product NGLQ operator, we determine the tightest independent rectangular bounding box for the truncated spectral support:
\begin{align}
	k_x \in \left[-\frac{2\pi}{\lambda}\sin \beta_{\max}, \frac{2\pi}{\lambda}\sin \beta_{\max}\right],\nonumber \\
	k_y \in \left[-\frac{2\pi}{\lambda}\sin\alpha_{\max}, \frac{2\pi}{\lambda}\sin\alpha_{\max}\right].
	\label{eq:k_x_k_y_limited}
\end{align}
Note that the maximum absolute value of $k_x$ is attained at $\alpha = 0$ and $|\beta| = \beta_{\max}$, while that of $k_y$ is attained at $|\alpha| = \alpha_{\max}$.
This environment-aware wavenumber domain mapping is illustrated in Fig.~\ref{Fig:Limited_wavenumber_domain}.
\begin{figure}[!t]
	\centerline{\includegraphics[scale=0.1]{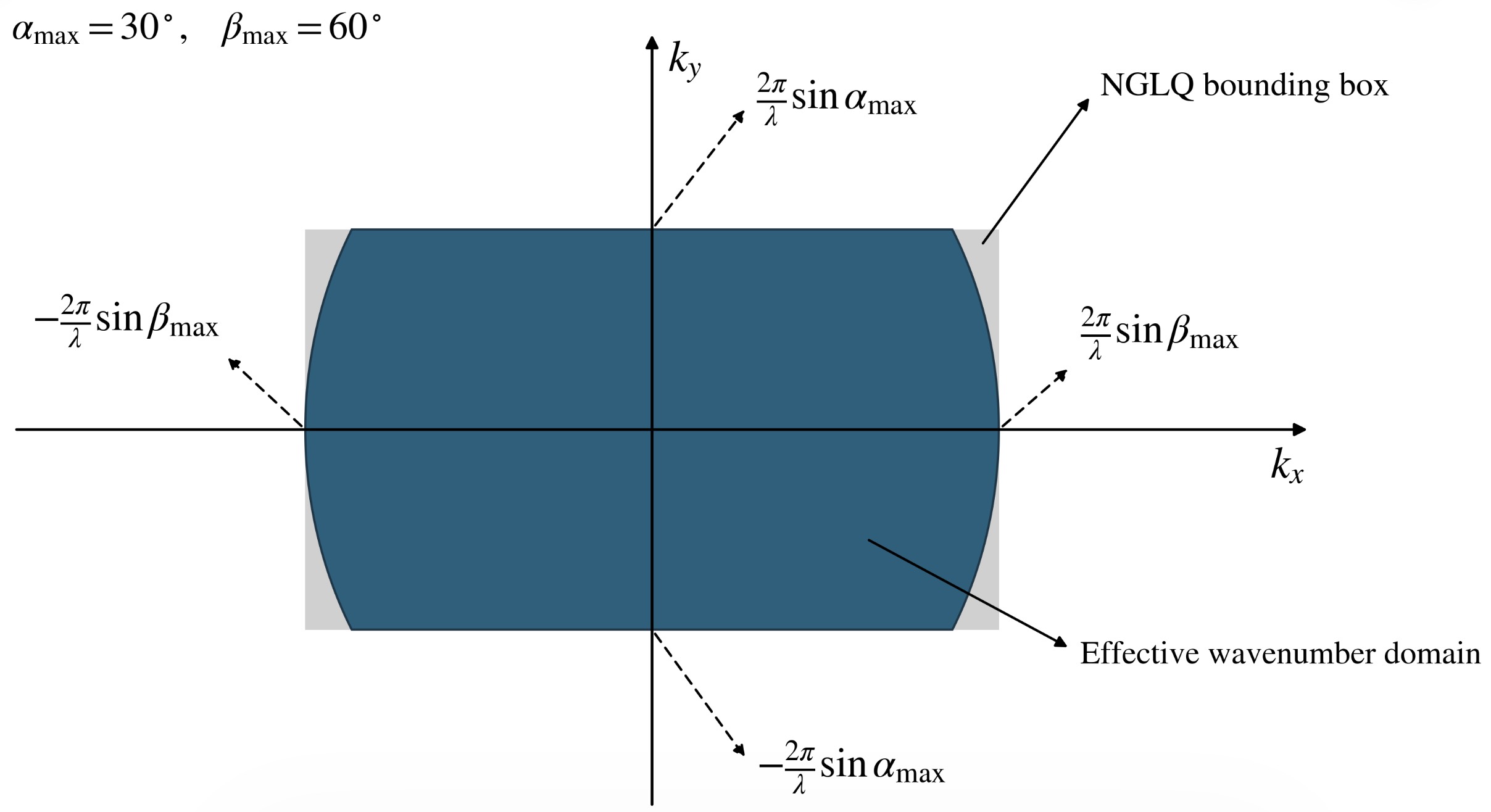}}
	\caption{Illustration of the effective wavenumber domain and its corresponding NGLQ bounding box with $\alpha\in[-\alpha_{\max}, \alpha_{\max}]$ and $\beta \in [-\beta_{\max}, \beta_{\max}]$.}
	\label{Fig:Limited_wavenumber_domain}
\end{figure}
Following the mathematical framework established in Theorem~\ref{theorem:2D_convergence}, a sufficient GLQ node requirement for super-exponential convergence is given by:
\begin{equation}
	M_x  > \frac{\pi L_x}{\lambda}\sin\beta_{\max} \quad
	\text{and} \quad M_y > \frac{\pi L_y}{\lambda}\sin\alpha_{\max}.
\end{equation}
Consequently, restricting the horizontal and vertical angular coverage (e.g., in sectorized cellular networks) can effectively reduce the corresponding GLQ node requirements.

\section{Effective DoF and Numerically Stable Regime}
\label{sec:eDoF}
In this section, we first evaluate the numerical reconstruction error (RE) of the sinc kernel $\frac{\sin(c(s-t))}{\pi(s-t)}$, and then provide a detailed analysis of the numerical stability of applying NGLQ for kernel reconstruction. Building upon this analysis, we propose an aperture-size-dependent numerically stable regime for determining the number of reliably retained eigenmodes under finite numerical precision.

We define the empirical $\text{RE}(s)$ for evaluating the continuous kernel reconstruction at a predefined reference observation point $s$ (e.g., $s = -1$ at the boundary or $s = 0$ at the center of the normalized interval $[-1, 1]$) over a dense grid of test points as:
\begin{equation}
	\text{RE}(s) = \sum_{i=1}^{N_{\text{test}}}\Big|\bar{K}(s, t_i) - \hat{K}(s, t_i)\Big|/N_{\text{test}},
	\label{eq:RE_definition}
\end{equation}
where $\hat{K}(s,t_i)$ denotes the kernel reconstructed by the NGLQ method at the $i$-th test point $t_i$, and $N_{\text{test}}$ is the total number of evaluation points. 
We note that $s$ serves as a parameter in the quadrature analysis. Although the numerical value of $\mathrm{RE}(s)$ may vary slightly with the observation point, the cDoF threshold and the associated super-exponential convergence behavior are determined by the wavenumber support with respect to the integration variable $t$ and therefore remain unchanged as $s$ varies.
The numerical $\text{RE}(s)$ curves for the sinc kernel under various normalized linear aperture sizes are presented in Fig.~\ref{Fig:RE_Sinc_1D}.
\begin{figure}[!t]
	\centerline{\includegraphics[scale=0.5]{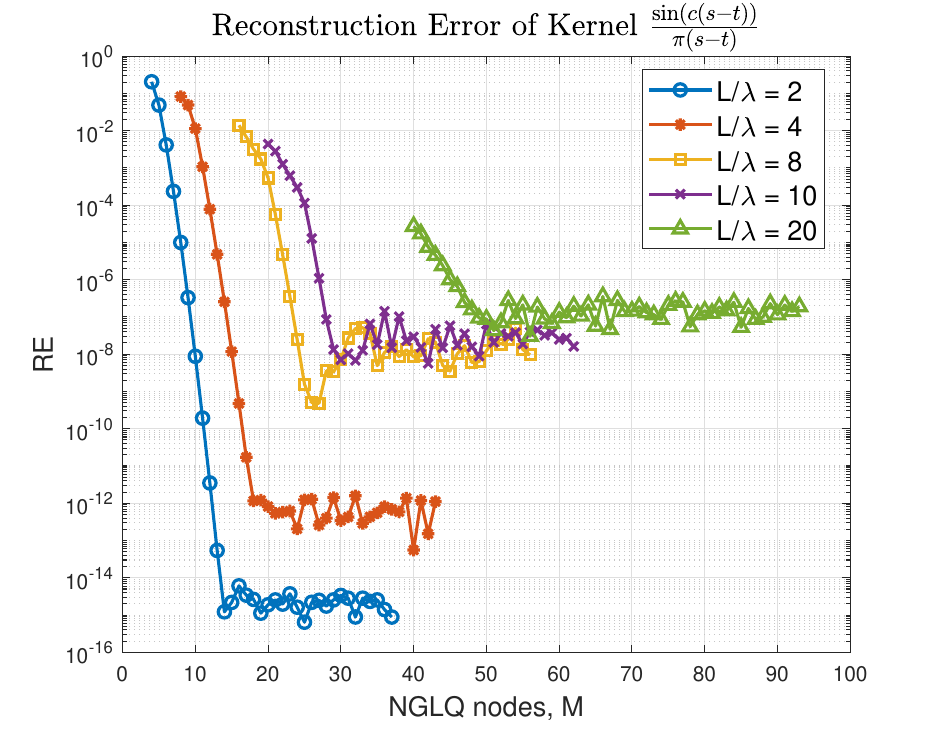}}
	\caption{$\text{RE}(s)$ of $\frac{\sin(c(s-t))}{\pi(s-t)}$ against NGLQ nodes under various normalized linear aperture sizes at $s=-1$.}
	\label{Fig:RE_Sinc_1D}
	\vspace{-0.4cm}
\end{figure}
For all considered linear aperture sizes, the number of NGLQ nodes $M$ starts from the physical DoF $2L/\lambda$. 

We first note that the isotropic kernel includes a continuous spectrum of spatial-frequency components, with lower-frequency components converging before the worst-case threshold $\pi L/\lambda$, which explains the initial RE decay at smaller $M$. The cDoF is the sufficient threshold ensuring that the highest admissible spatial-frequency component enters the super-exponential convergence regime, rather than the minimum node count required for every particular channel spectrum.

However, as observed in Fig.~\ref{Fig:RE_Sinc_1D}, a critical phenomenon is that the RE eventually hits an aperture-dependent error floor when all $M$ eigenmodes are retained, i.e., $N_{\text{trunc}} = M$. This phenomenon stems from a conflict between the super-exponential convergence threshold ($M \approx \pi L/\lambda$) and the numerical stability of the Nystr\"{o}m interpolation in (\ref{eq:eigenfunction_interpolation}). As $M$ increases, the higher-order eigenvalues of the Hermitian Nystr\"{o}m matrix $\mathbf{W}^{1/2}\mathbf{K}\mathbf{W}^{1/2}$ in \eqref{eq:Nystrom_EVD_Symmetric} decay toward the finite-precision level. Since the continuous eigenfunction reconstruction in \eqref{eq:eigenfunction_interpolation} involves division by $\tilde{\mu}_n$, retaining these near-zero modes amplifies numerical errors and eventually produces the observed error floor. For a physical covariance operator, we define its eDoF under a prescribed relative eigenvalue threshold $\epsilon\in(0, 1)$ as the number of eigenvalues satisfying $\mu_n/\mu_1 \geq \epsilon$. This eDoF can be viewed as the finite effective rank of the corresponding physical covariance and is invariant to overall covariance scaling.
In the following subsections, we first characterize the asymptotic eDoF reference for linear apertures and then extend it to 2D rectangular apertures through a semi-analytical formulation.
Subsequently, we establish a numerically stable regime for reconstructing the channel or kernels.

\subsection{The eDoF Reference for Linear Arrays}
\label{sec:the_eDoF_reference_for_linear_arrays}
For both linear and rectangular arrays, we characterize the asymptotic eDoF reference by invoking the Szeg\H{o}-Widom asymptotic expansion~\cite{Widom1982, sobolev2013pseudo, Sobolev2017Functions}.
The eigenvalues of the concentration operator lie in $(0,1)$, with the largest approaching unity as the space-bandwidth product increases. Accordingly, the threshold $\epsilon$ used in the concentration-operator references below asymptotically corresponds to the relative eigenvalue threshold defined above.
Specifically, for the 1D case, by analyzing the trace of the 1D concentration operator, we show that the 1D eDoF reference, $N_{\text{eDoF}}^{(1D)}(\epsilon)$, has the following approximate analytical form:
\begin{equation}
	N_{\text{eDoF}}^{(1D)}(\epsilon)  \approx \left\lceil \frac{2L}{\lambda} + \frac{1}{\pi^2}\ln\left(\frac{1-\epsilon}{\epsilon}\right)\ln\left(\pi\frac{L}{\lambda}\right)\right\rceil.
	\label{eq:N_effect_rank}
\end{equation}
Physically, this asymptotic expression decomposes the 1D eDoF reference into the classical length-law $2L/\lambda$ and a logarithmic edge-correction term capturing the finite aperture truncation. The detailed derivation is provided in Appendices~\ref{appdix:Widom_conjecture} and \ref{sec:1D_proof}. We note that the approximation in (\ref{eq:N_effect_rank}) aligns with the proven classic result in~\cite{landau1980eigenvalue}. 

Using (\ref{eq:N_effect_rank}) and setting $\epsilon = 10^{-16}$, the values of $N_{\text{eDoF}}^{(1D)}(\epsilon)$ for the normalized aperture sizes considered in Fig. \ref{Fig:RE_Sinc_1D} are 11, 18, 29, 33, and 56, which are close to the number of NGLQ nodes yielding the smallest RE values. 
We emphasize that~\eqref{eq:N_effect_rank} characterizes the asymptotic eDoF of the concentration operator $\mathcal{T}_{\alpha}=P_{\alpha,\Omega}Q_\Gamma P_{\alpha,\Omega}$ defined in \eqref{eq:Operator_T_alpha} for the 1D case, i.e., $d=1$. Its corresponding spatial kernel is $\frac{\sin(c(s-t))}{\pi(s-t)}$. Since this operator uses the full admissible wavenumber support with uniform spectral weighting, its leading term $2L/\lambda$ corresponds to the full-support phase-space limit~\cite{landau1980eigenvalue, Widom1982}. 
Thus, \eqref{eq:N_effect_rank} provides an a priori eigenvalue-count reference determined by the aperture, admissible wavenumber support, and threshold $\epsilon$, including the boundary correction for finite-aperture truncation.
For a general non-isotropic scattering environment, the spectral weighting within the admissible wavenumber support may redistribute the eigenvalue spectrum.

\subsection{The eDoF Reference for 2D Rectangular Arrays}
As discussed in Section \ref{sec:Introduction}, the multidimensional Szeg\H{o}-Widom asymptotic expansion for the specific setting considered here, involving a rectangular spatial domain and the discontinuous step test function used for eigenvalue counting, has not been rigorously established. Therefore, we propose a semi-analytical asymptotic eDoF reference for 2D rectangular arrays in this subsection, which is further assessed through numerical results.

\begin{result}[Semi-analytical eDoF Reference for 2D Planar HMIMO]
	For a continuous planar rectangular aperture with physical dimensions $L_x \times L_y$ and the full propagating wavenumber disk of radius $\kappa = 2\pi/\lambda$, consider the concentration operator $\mathcal{T}_\alpha$ defined in \eqref{eq:Operator_T_alpha} for the 2D case, i.e., $d = 2$. For a prescribed threshold $\epsilon \in (0, 1)$ on the concentration-operator eigenvalue scale, the corresponding asymptotic eDoF reference is approximated as
	\begin{align}
		&N_{\text{eDoF}}^{(2D)}(\epsilon) \approx \left\lceil\frac{\pi A_{\text{array}}}{\lambda^2} + \right. \nonumber \\
		& \; \left. \frac{2}{\pi^2}\ln\left( \frac{1-\epsilon}{\epsilon}\right) \left[\frac{L_y}{\lambda}\ln\left(\pi\frac{L_x}{\lambda}\right)+\frac{L_x}{\lambda}\ln\left(\pi \frac{L_y}{\lambda}\right)\right]\right\rceil,
		\label{eq:N_rank_2D}
	\end{align}
	where $A_{\text{array}} = L_xL_y$ is the physical area of the aperture.
\end{result}
The derivation in Appendices~\ref{appdix:Widom_conjecture} and~\ref{sec:2D_proof} follows the 1D Szeg\H{o}-Widom framework by establishing the operator trace model, evaluating the dominant volume term, and decoupling the boundary contributions along the two orthogonal dimensions of the rectangular aperture.

\begin{remark}
	Equation (\ref{eq:N_rank_2D}) decomposes the asymptotic eDoF reference into two distinct terms. The first term, $\frac{\pi A_{\text{array}}}{\lambda^2}$, represents the classical area-law DoF proportional to the normalized area, neglecting boundary effects. The second, logarithmic term captures the finite-size truncation effect at the rectangular boundaries. Notably, this edge correction is cross-coupled: the normalized length of a boundary (e.g., $L_y/\lambda$) determines the magnitude of the edge effect, while the spatial truncation along its orthogonal axis dictates the logarithmic space-bandwidth penalty (e.g., $\ln(\pi L_x/\lambda)$).
\end{remark}

To clarify the relationship between the concentration reference in \eqref{eq:N_rank_2D} and the physical covariance under 3D isotropic scattering, we first specify the two kernels involved in the comparison of Fig.~\ref{Fig:Concentration_Isotropic_Eigenspectrum}. According to \eqref{eq:Q_Gamma} and \eqref{eq:kappa_r}, for the 2D propagating wavenumber disk $k_x^2 + k_y^2 \leq \kappa^2$, the kernel associated with the concentration operator $\mathcal{T}_{\alpha}$ is given by
\begin{align}
	K_{\mathrm{con}}(\mathbf{r}, \mathbf{r}') =
	\frac{1}{4\pi^2}
	\iint_{k_x^2 + k_y^2 \leq \kappa^2} e^{j(k_x(x-x') + k_y(y-y'))}
	\, dk_x \, dk_y. \nonumber \\
	\label{eq:K_con}
\end{align}
For the physical covariance under 3D isotropic scattering, the spectral factor in \eqref{eq:h_pm_x_y_z} is direction independent, while the spherical-to-planar projection introduces the weighting factor \((\kappa^2-k_x^2-k_y^2)^{-1/2}\). Since the aperture is planar, the upgoing and downgoing components yield the same spatial covariance dependence. Therefore, after normalization to unit value at zero spatial separation, the isotropic covariance kernel is
\begin{align}
	K_{\mathrm{iso}}
	(\mathbf{r}, \mathbf{r}')
	&=
	\frac{1}{2\pi\kappa}
	\iint_{k_x^2 + k_y^2 \leq \kappa^2}
	\frac{e^{j(k_x(x-x') + k_y(y-y'))}}
	{\sqrt{\kappa^2 - k_x^2 - k_y^2}}
	dk_x \, dk_y.
	\label{eq:K_iso}
\end{align}
We then have the one-sided eigenvalue comparison result in the following lemma.
\begin{lemma}
	\label{lemma:isotropic_concentration_ordering}
	Let $\mathcal{T}_{\mathrm{con}}$ and $\mathcal{T}_{\mathrm{iso}}$ denote the integral operators generated by $K_{\mathrm{con}}(\mathbf{r}, \mathbf{r}')$ and $K_{\mathrm{iso}}(\mathbf{r}, \mathbf{r}')$, respectively, over the same planar aperture. Then 
	$
	\mathcal{T}_{\mathrm{iso}} \succeq \frac{\lambda^2}{2\pi}\mathcal{T}_{\mathrm{con}},
	$
	where $\succeq$ denotes the positive-semidefinite operator ordering. Consequently, if $\mu_n^{\mathrm{con}}$ and $\mu_n^{\mathrm{iso}}$ denote their ordered eigenvalues, respectively, we have
	\begin{equation}
		\bar{\mu}_n^{\mathrm{iso}}
		\geq
		\frac{\lambda^2\mu_1^{\mathrm{con}}}{2\pi\mu_1^{\mathrm{iso}}}\bar{\mu}_n^{\mathrm{con}},
		\label{eq:one-side_ordering}
	\end{equation}
	where $\bar{\mu}_n^{\mathrm{iso}} = \mu_n^{\mathrm{iso}}/\mu_1^{\mathrm{iso}}$ and $\bar{\mu}_n^{\mathrm{con}} = \mu_n^{\mathrm{con}}/\mu_1^{\mathrm{con}}$.
\end{lemma}
The proof of Lemma~\ref{lemma:isotropic_concentration_ordering} is provided in Appendix~\ref{appendix:eigenvalue_ordering}.

Lemma~\ref{lemma:isotropic_concentration_ordering} establishes a rigorous one-sided relation between the two normalized eigenvalue distributions. A reverse operator ordering, however, does not follow from the corresponding spectral weightings because the isotropic weighting becomes singular near the boundary of the propagating disk. We therefore examine the reverse deep-tail relationship numerically.
Fig.~\ref{Fig:Concentration_Isotropic_Eigenspectrum} compares the two normalized eigenvalue distributions $\{{\bar{\mu}}_n^{\mathrm{iso}}\}$ and $\{\bar{\mu}_n^{\mathrm{con}}\}$ using the same NGLQ discretization. For all considered apertures, the normalized physical-isotropic eigenvalues remain below those of the concentration operator in the deep-tail region relevant to numerical truncation. The two distributions also exhibit comparable decay behavior in this regime. These observations provide numerical evidence that the concentration spectrum serves as a conservative reference for the isotropic covariance in the truncation-relevant deep tail.
\begin{figure}[!t]
	\centerline{\includegraphics[scale=0.42]{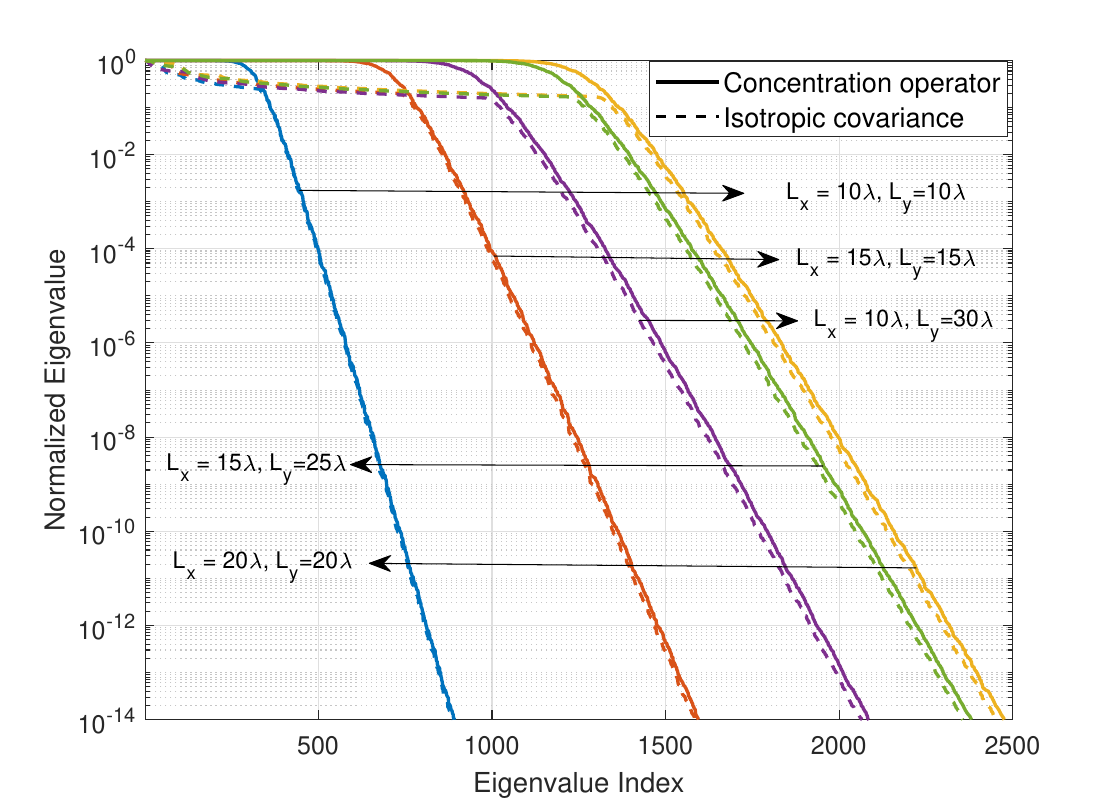}}
	\caption{Normalized eigenvalue distribution comparison between the 2D concentration operator and the physical isotropic covariance for representative square and rectangular apertures.}
	\label{Fig:Concentration_Isotropic_Eigenspectrum}
\end{figure}

To assess the semi-analytical reference in \eqref{eq:N_rank_2D}, we compare it against the numerical results obtained via the NGLQ method in Figs.~\ref{Fig:Eigvalue_Distri_Compar} and~\ref{Fig:Eigvalue_Distri_Compar_Diff_Lx_Ly}. Specifically, for a set of prescribed eigenvalue thresholds $\epsilon$, \eqref{eq:N_rank_2D} directly provides the corresponding eigenvalue counts $N_{\mathrm{eDoF}}^{(2D)}(\epsilon)$. The resulting pairs $(N_{\mathrm{eDoF}}^{(2D)}(\epsilon), \epsilon)$ provide a parametric representation of the semi-analytical eigenvalue distribution.
\begin{figure}[!t]
	\centerline{\includegraphics[scale=0.41]{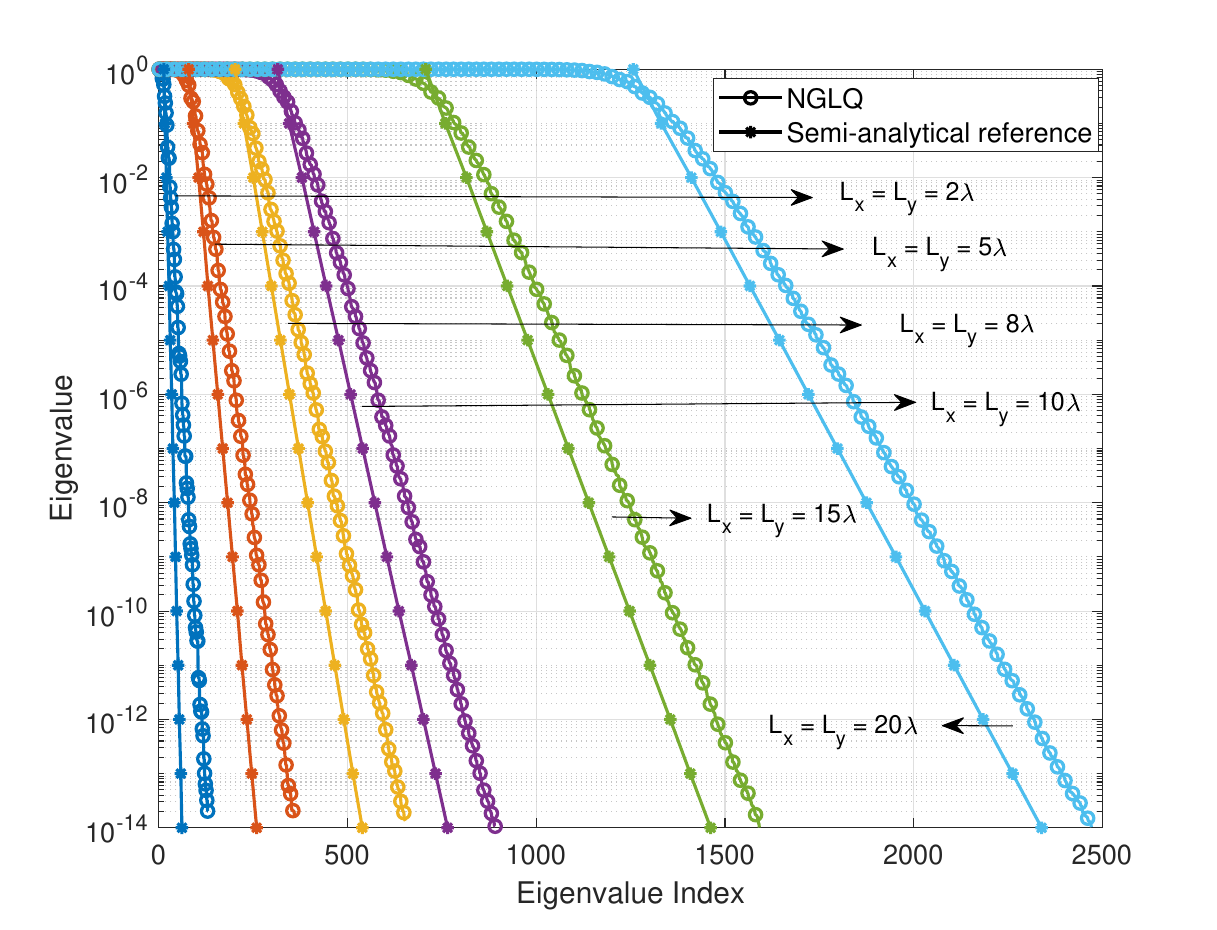}}
	\caption{Eigenvalue distribution of the concentration operator obtained via NGLQ and the corresponding semi-analytical reference for square apertures of different sizes.}
	\label{Fig:Eigvalue_Distri_Compar}
\end{figure}
Fig.~\ref{Fig:Eigvalue_Distri_Compar} presents the comparison for square apertures ($L_x = L_y$) of varying sizes. We denote the eDoF obtained numerically from the concentration operator as $N_{\text{eDoF}}^{\text{NGLQ}}(\epsilon)$, and define the relative gap between the numerical count and the semi-analytical reference $N_{\text{eDoF}}^{(2D)}(\epsilon)$ as $\text{RGAP}(\epsilon) = N_{\text{eDoF}}^{\text{NGLQ}}(\epsilon)/N_{\text{eDoF}}^{(2D)}(\epsilon) - 1$. 
Given a fixed threshold (e.g., $10^{-4}$), the relative discrepancy decreases as the normalized aperture size increases. We note that $\text{RGAP}(10^{-4})$ drops from 37.2\% at $L_x = L_y = 2\lambda$ to just 5.5\% at $L_x = L_y = 20\lambda$, indicating the progressively improved accuracy of the asymptotic reference in the large-aperture regime.
\begin{figure}[!t]
	\centerline{\includegraphics[scale=0.38]{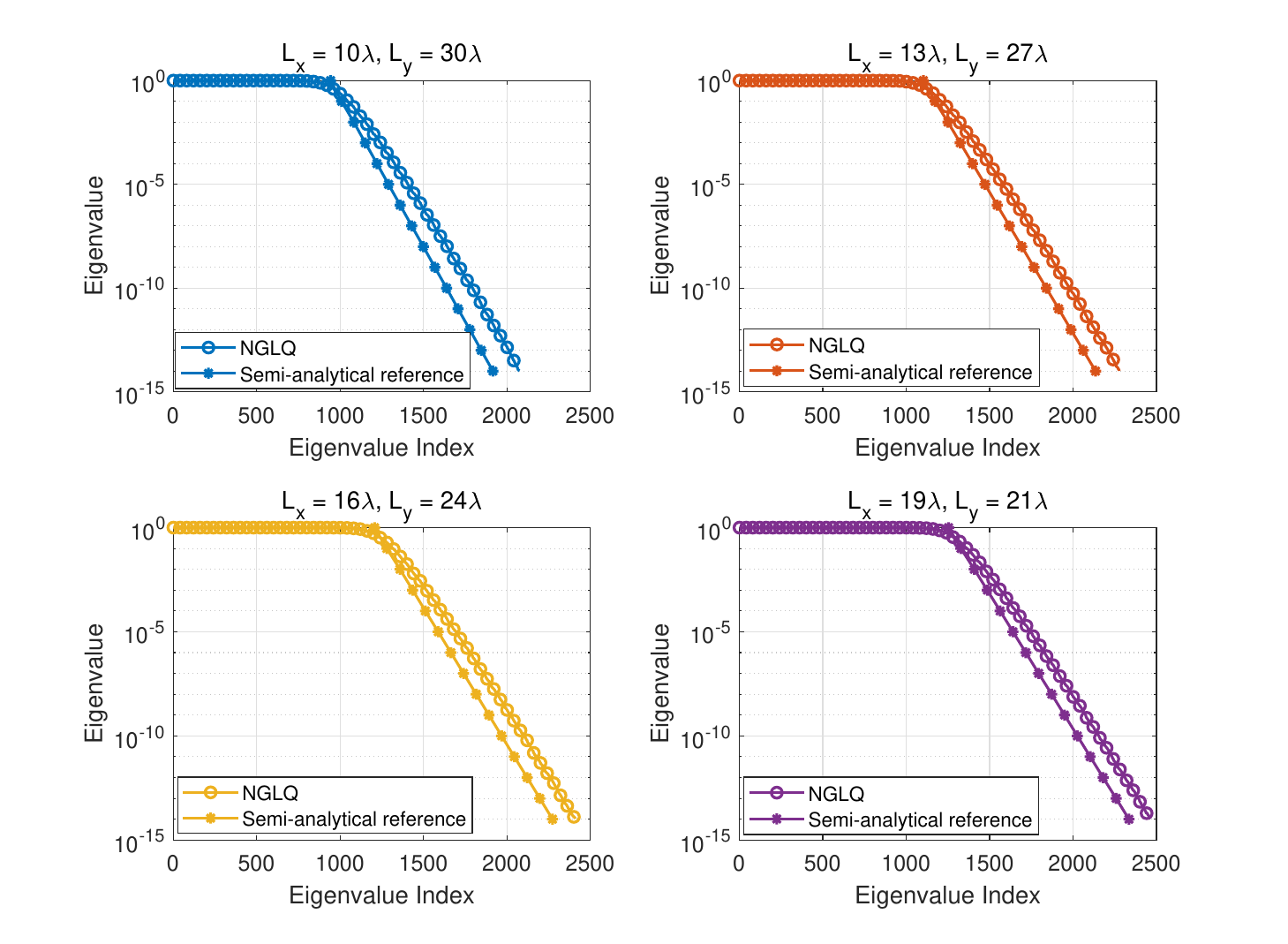}}
	\caption{Eigenvalue distribution of the concentration operator obtained via NGLQ and the corresponding semi-analytical reference for rectangular apertures with different aspect ratios.}
	\label{Fig:Eigvalue_Distri_Compar_Diff_Lx_Ly}
	\end{figure} Fig.~\ref{Fig:Eigvalue_Distri_Compar_Diff_Lx_Ly} further investigates rectangular apertures with a fixed perimeter but varying aspect ratios. 
	It is observed that $\text{RGAP}(\epsilon)$ remains at a comparable level across different configurations, which validates the robustness of the approximation in (\ref{eq:N_rank_2D}). 
	Notably, the $\text{RGAP}(\epsilon)$ for $L_x = 10\lambda, L_y = 30\lambda$ is larger than that of $L_x = 19\lambda, L_y = 21 \lambda$, suggesting that the approximation accuracy is influenced by the smaller normalized aperture dimension, $\min(L_x/\lambda,L_y/\lambda)$.

\subsection{Numerically Stable Regime}
Based on the preceding analysis, we propose a numerically stable regime for low-complexity channel reconstruction.
This regime involves three key parameters: the total number of GLQ nodes $M$ (the 1D quadrature order for linear apertures and $M=M_xM_y$ for 2D rectangular apertures), the concentration-based eDoF reference $N_{\text{eDoF}}(\epsilon)$ for a prescribed eigenvalue threshold $\epsilon$, and the actual truncated number of eigenmodes $N_{\text{trunc}}$.
We first clarify the role of the threshold $\epsilon$ when transferring the concentration-based eDoF reference to a physical covariance operator. For the concentration operator, $N_{\text{eDoF}}^{(2D)}(\epsilon)$ counts the eigenvalues above $\epsilon$. 
As noted in Section \ref{sec:the_eDoF_reference_for_linear_arrays}, the corresponding spectral dynamic range $\mu_1^{\mathrm{con}}/\epsilon$ approaches $1/\epsilon$.
For a physical covariance operator whose dominant eigenvalue \(\mu_1^{\rm phy}\) generally differs from unity, the same relative spectral range corresponds to \([\epsilon\mu_1^{\rm phy},\mu_1^{\rm phy}]\).
Therefore, \(N_{\rm eDoF}^{(2D)}(\epsilon)\) serves as an \textit{a priori} reference for the number of numerically relevant eigenmodes of the physical covariance over the same spectral range.

For large normalized arrays, a full EVD of the high-dimensional kernel matrix is computationally expensive due to its \(\mathcal{O}(M^3)\) complexity. While partial eigenvalue solvers such as the Lanczos algorithm~\cite{lanczos1950iteration} are more efficient, they require an estimate of the number of dominant eigenpairs to be computed. Since the traditional physical DoF provides only a fixed geometric scale and does not account for a prescribed spectral threshold, we use $k_0 =\min\{M, N_{\mathrm{eDoF}}^{(2D)}(\epsilon)\}$ as the initial target dimension and increase it adaptively when necessary. Let \(k\) denote the final number of eigenpairs computed by the partial EVD. For \(k\ll M\), the computational complexity is reduced to approximately \(\mathcal{O}(M^2k)\).

After the leading \(k_0\) eigenpairs are computed, the dominant numerical eigenvalue \(\tilde{\mu}_1^{\rm phy}\) is directly available. The numerically relevant modes are identified using $ \tilde{\mu}_n^{\rm phy}\ge \epsilon\tilde{\mu}_1^{\rm phy}.$
If the smallest computed eigenvalue still satisfies this criterion, the target dimension is increased until the computed eigenspectrum extends below \(\epsilon\tilde{\mu}_1^{\rm phy}\) or $k = M$. The number of eigenvalues satisfying the criterion defines \(N_{\rm trunc}\), and the corresponding eigenmodes are retained for channel reconstruction.

\section{Validation under Non-isotropic Scattering}
\label{sec:validation_under_NISO}
Unlike the isotropic case, the spatial correlation kernel under non-isotropic scattering generally lacks a tractable closed-form expression. To validate the NGLQ framework under non-isotropic scattering, we adopt the widely utilized von Mises-Fisher (vMF) mixture model~\cite{Pizzo2022Spatial, Wang2022Electromagnetic,Yan2026Spectrally} to describe the normalized squared spectral factor, $\bar{A}_h^2(\theta, \phi)$, on the unit sphere with $\bar{k}_x = \sin\theta \cos\phi$, $\bar{k}_y = \sin\theta\sin\phi$, and $\bar{k}_z = \cos\theta$. Assuming $N_c$ scattering clusters, the normalized squared spectral factor is modeled as 
\begin{equation}
	\bar{A}_h^2(\theta, \phi) = \sum_{j=1}^{N_c} v_j p_j(\theta, \phi)
	\label{eq:A_theta_phi}
\end{equation}
where $v_j > 0$ are mixture weights satisfying $\sum_{j=1}^{N_c} v_j = 1$. The vMF distribution $p_j(\theta, \phi)$ associated with the $j$-th cluster is given by:
\begin{align}
	&p_j(\theta, \phi) = \frac{\rho_j}{4\pi\sinh\rho_j}\times \nonumber \\
	& \exp\left\{\rho_j\left[\sin\theta\sin\mu_{\theta, j}\cos(\phi - \mu_{\phi, j}) + \cos\theta\cos\mu_{\theta, j}\right]\right\},
	\label{eq:p_theta_phi}
\end{align}
where $\mu_{\theta, j}$ and $\mu_{\phi, j}$ represent the modal cluster directions, and $\rho_j$ is the concentration parameter controlling the angular power spread. Following \cite{Pizzo2022Spatial}, we restrict the propagation to the upper hemisphere, i.e., $\theta \in [0, \pi/2)$, by neglecting contributions from scatterers located behind the receiver. The corresponding spatial correlation kernel $K(\Delta x, \Delta y)$, where \(\Delta x\) and \(\Delta y\) denote the spatial separations along the \(x\)- and \(y\)-axes, respectively, can then be expressed directly in angular coordinates as
 \begin{align}
 	&K(\Delta x, \Delta y)
 	\propto \nonumber \\
 	&\int_0^{2\pi}\int_0^{\pi/2}
 	\bar{A}_{h}^2
 	(
 	\theta, 
 	\phi
 	)
 	e^{
 	j\kappa
 	\sin\theta
 	(
 	\Delta x\cos\phi 
 	+
 	\Delta y\sin\phi
 	)
 	}
 	\sin\theta
 	\, d\theta
 	\, d\phi,
 	\label{eq:K_bar_x_bar_y}
 \end{align}

The above vMF model characterizes the diffuse scattering component. A Rician channel can be incorporated by superimposing a deterministic line-of-sight (LoS) plane wave on the vMF-based field, i.e.,
\begin{equation}
	h_{\text{Ric}}(\mathbf{r}) =
	\sqrt{
		\frac{\eta_R}{\eta_R + 1}
		}
		e^{j
			\mathbf{k}_0^T \mathbf{r}
			}
	+
	\sqrt{
	\frac{1}{\eta_R + 1}
	} 
	h_{\text{vMF}}(\mathbf{r}),
\end{equation}
where $\eta_R$ is the Rician factor, $\|\mathbf{k}_0\|_2 = \kappa$, and $h_{\text{vMF}}(\mathbf{r})$ denotes the zero-mean diffuse field generated from the vMF mixture in \eqref{eq:A_theta_phi} and can be reconstructed using the procedure given in Section~\ref{sec:NGLQ}. The vMF-mixture representation is also structurally compatible with cluster-based 3GPP-type channel models~\cite{3GPP38901}, where cluster directions, angular spreads, and relative powers can be represented through the corresponding vMF parameters.

Although the present analysis focuses on a single spatially-stationary link, the proposed continuous-to-discrete framework can also serve as a building block for multi-user channel modeling. User-specific vMF mixtures may be constructed by using spatially correlated cluster parameters following the spatial-consistency principles adopted in 3GPP channel models~\cite{3GPP38901}. Each resulting user-specific covariance kernel can then be discretized using the same NGLQ procedure. Explicitly modeling joint cross-user small-scale fading statistics is a worthwhile direction for future research.

To apply the NGLQ method, the continuous kernel must be evaluated to construct the $M \times M$ discrete kernel matrix $\mathbf{K}$ across all pairwise spatial differences ($\Delta x = x_m - x_{m'}$, $\Delta y = y_n - y_{n'}$) between the $M$ planar grid nodes.
If a physically admissible spatial covariance kernel $K$ is directly available, it can be used as the input to NGLQ.
However, since a closed-form expression for \eqref{eq:K_bar_x_bar_y} is generally unavailable under arbitrary vMF non-isotropic scattering, these kernel values must be obtained numerically. Conventionally, evaluating the continuous spatial kernel relies on the IDFT with extensive zero-padding, followed by cubic interpolation~\cite{Yan2026Spectrally}, which introduces wavenumber-domain discretization and spatial-interpolation errors that may obscure the high-accuracy NGLQ reconstruction behavior.
We therefore directly evaluate \eqref{eq:K_bar_x_bar_y} in the angular domain using high-order GLQ along $\theta$ and a periodic trapezoidal rule along $\phi$. The corresponding quadrature orders are progressively increased until further refinement produces negligible changes in the evaluated kernel over the spatial region of interest.

Direct evaluation over all pairwise spatial differences incurs a cost of approximately \(\mathcal{O}(M^2Q)\), where \(Q\) is the total number of angular quadrature points. For the validation scenarios considered here, this high-accuracy evaluation remains computationally manageable. For large-aperture implementations, 
the kernel-evaluation cost may be reduced for the vMF model by analytically eliminating the azimuthal integral \cite{Tubic2024Correlation}. Including eigenspace construction, the offline complexity of NGLQ is approximately $\mathcal{O}(M^2Q + M^2k)$ with partial EVD, where $k$ is the number of computed eigenpairs; a full EVD requires $\mathcal{O}(M^3)$. In contrast, FSE uses a prescribed Fourier basis and avoids EVD, but the powers of its $N_{\text{FSE}}$ modes must be obtained by integrating the angular spectrum over the corresponding wavenumber-grid cells, resulting in approximately $\mathcal{O}(N_{\text{FSE}}q_{\text{cell}})$ for $q_{\text{cell}}$ quadrature evaluations per cell. For the considered aperture regime, $M$ is generally larger than $N_{\text{FSE}}$, while $Q$ and $q_{\text{cell}}$ are accuracy-dependent quadrature parameters and therefore have no fixed ordering. After preprocessing, channel synthesis over $N_s$ spatial points requires $\mathcal{O}(N_s N_{\text{trunc}})$ operations for NGLQ and $\mathcal{O}(N_s N_{\text{FSE}})$ for direct FSE synthesis. Since both methods construct channel representations offline, when their preprocessing costs are affordable, representation accuracy becomes the more relevant criterion for the comparison that follows.

\subsection{Validation Results}
\begin{figure}[!t]
	\centerline{\includegraphics[scale=0.39]{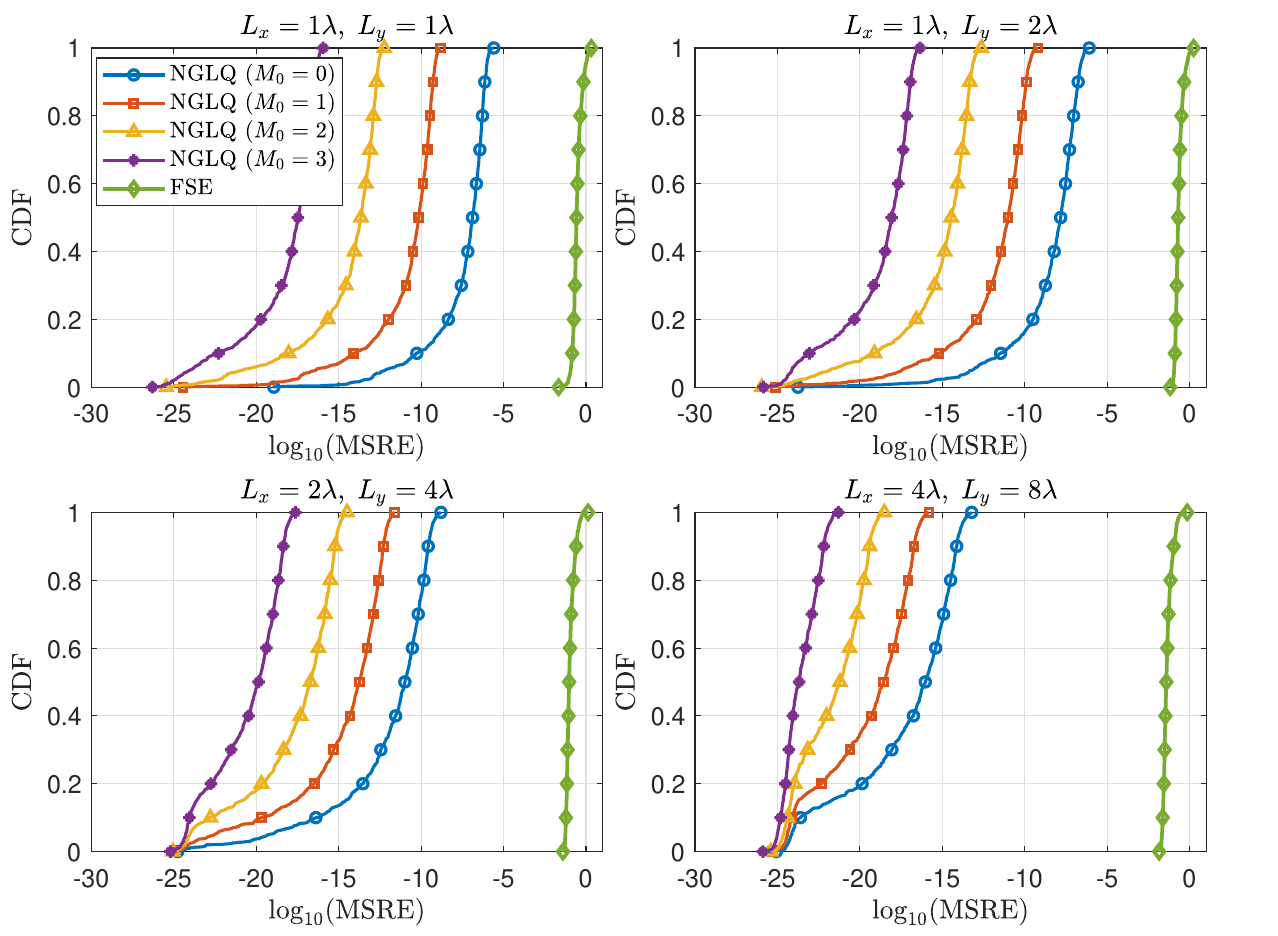}}
	\caption{CDF of kernel-reconstruction MSRE under non-isotropic scattering for various normalized aperture sizes. The reference spatial kernel is evaluated via angular quadrature.}
	\label{Fig:CDF_MSRE_NGLQ_Non-iso_angular_quadrature_Diff_Lx_Ly}
\end{figure}
\begin{figure}[!t]
	\centerline{\includegraphics[scale=0.36]{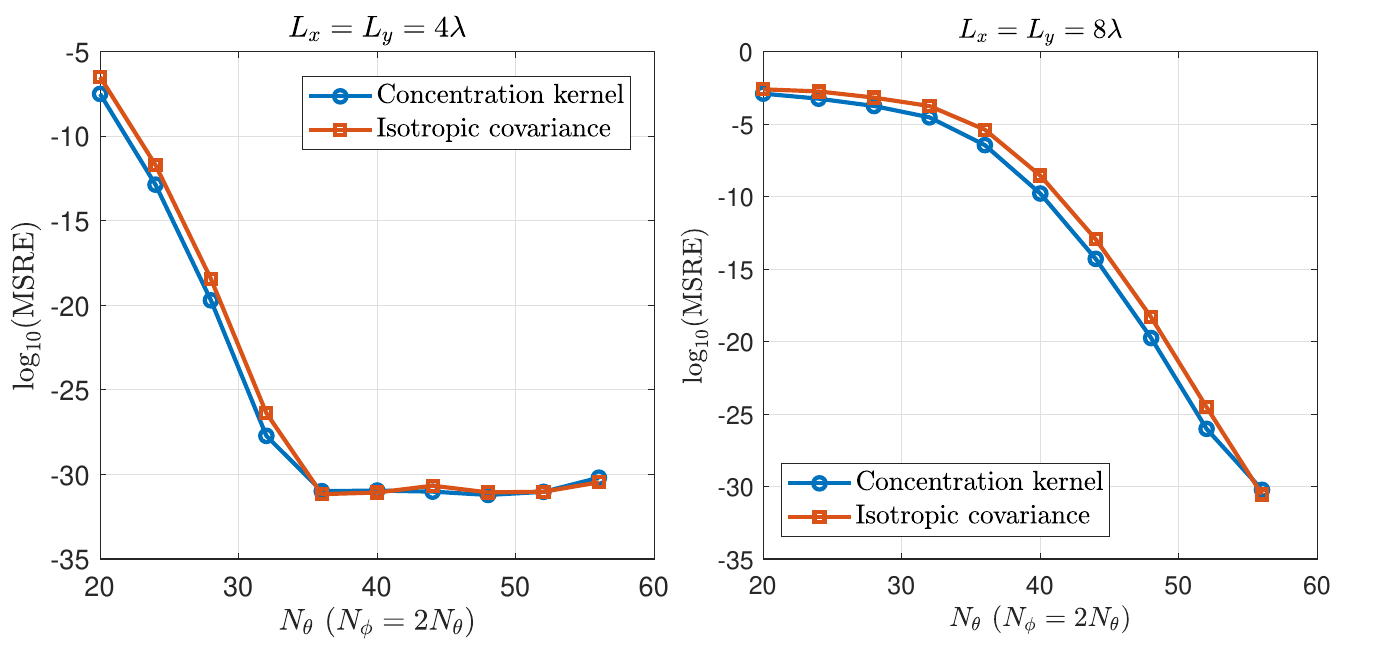}}
	\caption{Convergence of the angular quadrature against the closed-form concentration-operator and isotropic-covariance kernels for $L_x=L_y=4\lambda$ and $8\lambda$, with $N_\phi=2N_\theta$.}
	\label{Fig:Angular_quadrature_convergence_check}
\end{figure}

In Fig.~\ref{Fig:CDF_MSRE_NGLQ_Non-iso_angular_quadrature_Diff_Lx_Ly}, we show the cumulative distribution function (CDF) of the mean squared reconstruction error (MSRE) between the reference spatial kernel evaluated via angular quadrature and the reconstructed kernel via the NGLQ method. Note that the discrete kernel matrix $\mathbf{K}$ is also constructed using the same angular-domain quadrature described above. 
Each curve in Fig.~\ref{Fig:CDF_MSRE_NGLQ_Non-iso_angular_quadrature_Diff_Lx_Ly} is evaluated over 1000 randomly generated vMF scattering realizations, i.e., realizations of $\bar{A}_h^2(k_x/\kappa, k_y/\kappa)$. 
To simulate $\bar{A}_h^2(k_x/\kappa, k_y/\kappa)$ using the vMF model, the number of clusters, $N_c$, is drawn from a discrete uniform distribution over $\{1, \ldots, 10\}$. 
The modal angles $\mu_{\theta, j}$ and $\mu_{\phi, j}$ are uniformly distributed over $[0, \pi/2)$ and $[0, 2\pi)$, respectively. 
The concentration parameter $\rho_j$ is uniformly distributed over $[1, 100]$, and each cluster weight, $v_j$, is initially sampled from a uniform distribution over $[0, 1]$ and subsequently normalized such that $\sum_{j=1}^{N_c}v_j = 1$. 

To independently verify the accuracy of the angular-domain kernel evaluation, Fig.~\ref{Fig:Angular_quadrature_convergence_check} evaluates the quadrature convergence against two kernels with closed-form references: the normalized 2D concentration-operator kernel (i.e., $2J_1(\kappa \Delta r)/(\kappa \Delta r)$, where $\Delta r = \sqrt{\Delta x^2 + \Delta y^2}$ denotes the spatial separation, and $J_1(\cdot)$ is the first-order Bessel function of the first kind) and the physical isotropic kernel (i.e., $\sin(\kappa \Delta r)/(\kappa \Delta r)$). The comparison is performed for $L_x = L_y = 4 \lambda$ and $L_x = L_y = 8\lambda$. We denote the polar and azimuthal quadrature orders as $N_{\theta}$ and $N_\phi$, respectively, and $N_\phi = 2N_\theta$ in Fig. \ref{Fig:Angular_quadrature_convergence_check}. As $N_\theta$ and $N_\phi$ increase, the MSRE of both kernels decreases rapidly, confirming the convergence of the adopted angular quadrature. 
For validation purposes, we use \(N_\theta = 56\) GLQ nodes and \(N_\phi = 112\) periodic trapezoidal quadrature nodes to sufficiently suppress the kernel-evaluation error in the non-isotropic evaluation of Fig.~\ref{Fig:CDF_MSRE_NGLQ_Non-iso_angular_quadrature_Diff_Lx_Ly}.

For the setup in Fig.~\ref{Fig:CDF_MSRE_NGLQ_Non-iso_angular_quadrature_Diff_Lx_Ly}, we set $M_x = \left\lceil \frac{\pi L_x}{\lambda}\right\rceil + M_0$ and $M_y = \left\lceil\frac{\pi L_y}{\lambda}\right\rceil + M_0$. 
For kernel reconstruction, we retain only the eigenfunctions whose normalized eigenvalues exceed $10^{-14}$. 
Under these conditions, the CDF of the MSRE shifts consistently toward lower error levels as $M_0$ increases, exhibiting a spectral decay of the NGLQ reconstruction error. 
The FSE results are also included as a reference for comparison.
It should be noted that $\pi L_x/\lambda$ and $\pi L_y/\lambda$ are the convergence thresholds for the maximum spatial frequency component $2\pi/\lambda$ along the $x$- and $y$-axes, respectively. 
For non-isotropic scattering, most contributing spatial-frequency components have projected wavenumber magnitudes below $2\pi/\lambda$.
Consequently, their convergence thresholds are smaller than $\pi L_x/\lambda$ and $\pi L_y/\lambda$. This explains why small MSRE values are observed even when $M_0 = 0$. For the extreme scenario (i.e., the end-fire scenario) shown in Section \ref{sec:Tightness of the cDoF Threshold}, it is observed that the MSRE enters the super-exponential decay region when $M_x > \pi L_x/\lambda$ and $M_y > \pi L_y/\lambda$.

\section{Conclusion}
\label{sec:conclusion}
In this paper, we have rigorously analyzed a spectral-order reconstruction method based on NGLQ to address the discretization challenges in continuous HMIMO channel modeling. We established the theoretical threshold for the cDoF associated with GLQ nodes and proved that exceeding this threshold drives the quadrature error into a super-exponential decay regime. Furthermore, by invoking the Szeg\H{o}-Widom asymptotic expansion, we developed a physically grounded semi-analytical eDoF reference for 2D rectangular apertures, capturing the anisotropic boundary truncation effects.
By unifying the concepts of cDoF and eDoF, we formulated a numerically stable regime that supports high-accuracy discrete modeling of HMIMO channels with reduced EVD cost.
Finally, we validated the spectral convergence behavior of the NGLQ framework under non-isotropic scattering through direct angular-domain evaluation of the spatial correlation kernel.
A direction for future work is the extension of this theoretical framework to near-field communication scenarios. Near-field propagation is characterized by spherical waves, presenting a serious theoretical challenge as the channel spatial stationarity is lost. Nevertheless, it is important to note that, provided that the non-radiative evanescent waves are neglected, the wavenumber domain of the near-field channel remains band-limited to the propagating disk. This underlying physical constraint implies that the $C^{\infty}$ smoothness of the continuous spatial kernel should be preserved, a crucial ingredient for spectral convergence. 

\appendices
\section{}
\label{appendix:super_exponential_threshold}
To rigorously analyze the convergence of $E_M(f)$, we first note that an $M$-point GLQ achieves an algebraic precision of degree $2M-1$. 
As will be established in the subsequent Lemma~\ref{lemma:frequency_doubling}, the standardized integrand $f(t)$ is an entire function, which guarantees it is infinitely differentiable ($C^{\infty}$) over the real axis. 
According to classical numerical integration theory~\cite{davis2007methods, hildebrand1987introduction}, the absolute truncation error for such a sufficiently smooth function on $[-1, 1]$ is strictly bounded by its $2M$-th derivative: 
\begin{equation}
	|E_M(f)|\leq \frac{2^{2M+1}(M!)^4}{(2M+1)[(2M)!]^3}\max_{\xi\in[-1,1]}|f^{(2M)}(\xi)|.
	\label{eq:E_N_f_upper_bound}
\end{equation}
While (\ref{eq:E_N_f_upper_bound}) provides an exact theoretical bound, its factorial structure obscures the explicit behavior with respect to $M$. 
To reveal the decay rate, we apply Robbins's double inequalities~\cite{robbins1955remark}: $\sqrt{2\pi M}\left(\frac{M}{e}\right)^M e^{\frac{1}{12M + 1}} < M! < \sqrt{2\pi M}\left(\frac{M}{e}\right)^M e^{\frac{1}{12M}}$.
By upper-bounding the numerator and lower-bounding the denominator in (\ref{eq:E_N_f_upper_bound}) using Robbins's bounds, we obtain the following strict analytical upper bound for $|E_M(f)|$:
\begin{equation}
	|E_M(f)| \leq C(M)\Big(\frac{e}{4M}\Big)^{2M}\max_{\xi\in[-1,1]}|f^{(2M)}(\xi)|,
	\label{eq:E_N_f_4}
\end{equation}
where $C(M) = \frac{\sqrt{\pi M}}{2M+1}\exp\left(\frac{1}{3M} - \frac{3}{24M+1}\right)$ is an algebraically decaying pre-factor, which is dominated by the super-exponential function $\big(\frac{e}{4M}\big)^{2M}$, and $e$ is Euler's number.

\subsection{Upper Bound of $f^{(2M)}(\xi)$}
We now consider the integrand $f(t)$. 
We first show that the band-limited kernel $K(x, x')$ and its eigenfunction $g(x)$ consist of continuous linear combinations of complex exponential functions.
\begin{lemma}
	\label{lemma:band_limit}
	For a spatially stationary 1D HMIMO channel operating under a maximum spatial wavenumber $\kappa = 2\pi/\lambda$, both the continuous kernel and its corresponding eigenfunctions, when mapped onto the standard interval $s, t\in [-1,1]$ as $\bar{K}(s,t)$ and $\bar{g}(t)$, are entirely composed of continuous linear combinations of complex exponential functions $e^{j \omega t}$, and their maximum angular frequency is bounded by $\omega_{\max} = \pi L/\lambda$.
\end{lemma}
\begin{proof}
	Substituting the standard affine transformations $x = \frac{L}{2}s + \frac{L}{2}$ and $x' = \frac{L}{2}t + \frac{L}{2}$ into the continuous kernel, its 1D plane-wave spectral representation on the standard interval becomes:
	\begin{equation}
		\bar{K}(s, t) = \int_{-\kappa}^{\kappa}S(k)e^{-jk\frac{L}{2}s}e^{jk\frac{L}{2}t}\,dk.
		\label{eq:tilde_K_s_t}
	\end{equation}
	By defining the effective spatial frequency $\omega = kL/2$, the boundary is naturally mapped to $\omega_{\max} = \kappa L/2 = \pi L/\lambda$. Since the differential scales as $dk = \frac{2}{L}d\omega$, the mapped kernel can be rewritten as a continuous superposition of $e^{j\omega t}$:
	\begin{equation}
		\bar{K}(s, t) = \int_{-\omega_{\max}}^{\omega_{\max}}W(s, \omega)e^{j\omega t}\,d\omega,
		\label{eq:tilde_K_s_t_2}
	\end{equation}
	where $W(s, \omega) = \frac{2}{L}S\Big(\frac{2\omega}{L}\Big)e^{-j\omega s}$. Thus, the kernel on the standard interval is band-limited with 
	no spatial frequency components higher than the physical bandwidth $\omega_{\max}$.
	
	To analyze the eigenfunctions, we recall the continuous KLD problem: $\frac{2\mu_n}{L} \bar{g}(s) = \int_{-1}^{1} \bar{K}(s, t)\bar{g}(t)dt$. Substituting the spectral decomposition of $\bar{K}(s, t)$ as in (\ref{eq:tilde_K_s_t}) and applying Fubini's theorem to exchange the order of integration yields:
	\begin{equation}
		\frac{2\mu_n}{L}\bar{g}(s) = \int_{-\kappa}^{\kappa}S(k)e^{-jk\frac{L}{2}s}\left[\int_{-1}^{1}\bar{g}(t)e^{jk\frac{L}{2}t}\,dt\right]\,dk.
	\end{equation}
	Recognizing the inner integral as the finite Fourier transform of $\bar{g}(t)$, denoted as $\phi(k) = \int_{-1}^{1}\bar{g}(t)e^{j k\frac{L}{2}t}\,dt$, and using the frequency mapping $\omega = k L/2$, the eigenfunction can similarly be expressed as a continuous superposition of $e^{-j\omega s}$:
	\begin{equation}
		\bar{g}(s) = \int_{-\omega_{\max}}^{\omega_{\max}}C(\omega)e^{-j\omega s}\,d\omega,
		\label{eq:tilde_g_s}
	\end{equation}
	where $C(\omega) = \frac{1}{\mu_n}S\left(\frac{2\omega}{L}\right)\phi\left(\frac{2\omega}{L}\right)$. This structural form demonstrates that the eigenfunctions are bounded by the exact same maximum spatial frequency as the kernel.
\end{proof}

Based on the bandlimited properties established in Lemma~\ref{lemma:band_limit}, we now characterize the growth and boundedness of the standardized integrand $f(t)$ to lay the groundwork for the high-order derivative analysis.
\begin{lemma}
	\label{lemma:frequency_doubling}
	The integrand $f(t) = \bar{K}(s, t) \bar{g}(t)$ is an entire function of exponential type $2\pi L/\lambda$.
\end{lemma}
\begin{proof}
According to the Paley-Wiener theorem~\cite{rudin1974real}, if a function $H(x)$ is supported in $[-\kappa, \kappa]$ such that $H \in L^2(-\kappa, \kappa)$, its holomorphic Fourier transform
\begin{equation}
	h(z) = \int_{-\kappa}^{\kappa}H(x)e^{jxz}\,dx
\end{equation}
is an entire function of exponential type $\kappa$, meaning there exists a constant $C$ such that $|h(z)| \leq Ce^{\kappa|z|}$ for all $z \in \mathbb{C}$. Based on Lemma~\ref{lemma:band_limit}, both the kernel function $\bar{K}(s, t)$ and the eigenfunction $\bar{g}(t)$ are bandlimited within $[-\omega_{\max}, \omega_{\max}]$, where $\omega_{\max} = \pi L/\lambda$. By analytically continuing $t$ to the complex plane $z \in \mathbb{C}$, both functions are classified as entire functions of exponential type $\omega_{\max} = \pi L/\lambda$. Consequently, they satisfy $|\bar{K}(s, z)| \leq C_K e^{\omega_{\max}|z|}$ and $|\bar{g}(z)| \leq C_g e^{\omega_{\max}|z|}$. Their product thus satisfies $|f(z)| \leq C_f e^{2\omega_{\max}|z|}$ (where $C_f = C_KC_g$), which proves that the integrand $f(t)$ is of exponential type $2\omega_{\max} = 2\pi L/\lambda$.
\end{proof}

To upper bound $f^{(2M)}(\xi)$, we resort to Bernstein's inequality for entire functions~\cite{Rahman2009Bernstein}. This classical theorem states that if an entire function $H(z)$ of exponential type $\tau$ is bounded by a constant $M_0$ on the real axis (i.e., $|H(x)| \leq M_0$ for $-\infty < x < \infty$), its derivative globally satisfies $|H'(x)| \leq M_0\tau$.

In practical HMIMO systems, the total radiated power is finite. This physical energy constraint guarantees that the kernel $\bar{K}(s, t)$ is bounded on the real spatial axis (i.e., $|\bar{K}| \le M_K$). Concurrently, the eigenfunction $\bar{g}(t)$ possesses finite physical energy over the restricted antenna aperture. Through the analytical continuation governed by the Fredholm integral equation, this local energy constraint intrinsically forces the eigenfunction to remain uniformly bounded across the entire real domain as well (i.e., $|\bar{g}| \leq M_g$). Therefore, the combined integrand is uniformly bounded by $|f(t)| \leq M_0 = M_KM_g$ for all $t \in \mathbb{R}$, satisfying the prerequisite for Bernstein's inequality.
	
Leveraging the exponential type property established in Lemma~\ref{lemma:frequency_doubling}, we iteratively apply Bernstein's inequality $2M$ times to establish an upper bound of $\max_{\xi \in [-1,1]}|f^{(2M)}(\xi)|$:
\begin{equation}
	\max_{\xi\in[-1,1]}|f^{(2M)}(\xi)| \leq \sup_{t \in \mathbb{R}}|f^{(2M)}(t)| \leq \left(\frac{2\pi L}{\lambda}\right)^{2M}\sup_{t\in\mathbb{R}}|f(t)|.
	\label{eq:upper_bound_f_2M}
\end{equation}
Substituting (\ref{eq:upper_bound_f_2M}) into (\ref{eq:E_N_f_4}), we obtain the analytical upper bound for $E_{M}(f)$, expressed in terms of the normalized aperture size $\frac{L}{\lambda}$ and the GLQ order $M$ as given in~(\ref{eq:envelope_E_M_f}).

\subsection{Convergence Threshold}
We note that (\ref{eq:envelope_E_M_f}) suggests a super-exponential decay of the quadrature error with respect to the number of GLQ nodes $M$. 
However, it serves only as an upper bound.
In this subsection, by shifting our perspective from the global integral operator to the local interpolation remainder, we show that a tighter convergence threshold can be established.

Assume we interpolate the integrand $f(t)$ over the interval $[-1, 1]$ using $M$ GLQ nodes. According to the fundamental Lagrange interpolation remainder theorem~\cite{hildebrand1987introduction}, the truncation error $\mathcal{R}_M(t)$ at any evaluation point $t$ is given by:
\begin{equation}
	\mathcal{R}_M(t) = \frac{f^{(M)}(\xi)}{M!}\prod_{m=1}^{M}(t-t_m),
	\label{eq:truncation_error}
\end{equation}
where $\xi \in [-1, 1]$.
Since $f(t)$ is an entire function of exponential type $2\pi L/\lambda$ bounded by $M_0$ on the entire real axis, we can obtain the upper bound for its $M$-th derivative via Bernstein's inequality, yielding $\sup_{\xi \in \mathbb{R}}|f^{(M)}(\xi)| \leq M_0 (2\pi L/\lambda)^M$. 
Furthermore, since the NGLQ method employs the roots of the Legendre polynomial as quadrature nodes, the nodal polynomial corresponds to the monic Legendre polynomial. Its maximum modulus over $[-1, 1]$ is attained at the boundaries, given by~\cite{hildebrand1987introduction}:
\begin{equation}
	\max_{t \in [-1,1]}\bigg|\prod_{m=1}^M(t-t_m)\bigg| = \frac{2^M(M!)^2}{(2M)!}.
\end{equation}
By substituting the above facts into (\ref{eq:truncation_error}), the absolute interpolation error $|\mathcal{R}_M(t)|$ is upper bounded by:
\begin{equation}
	|\mathcal{R}_M(t)| \leq M_0\frac{2^M M!}{(2M)!} \bigg(\frac{2\pi L}{\lambda}\bigg)^M.
	\label{eq:E_M_t_upper_bound}
\end{equation}
Let $U_M$ denote the upper bound sequence on the right-hand side of (\ref{eq:E_M_t_upper_bound}). By evaluating $U_{M+1}/U_M$, we obtain:
\begin{equation}
	\frac{U_{M+1}}{U_M} = \frac{1}{2M+1}\left(\frac{2\pi L}{\lambda}\right).
	\label{eq:ratio_test}
\end{equation}
Equation (\ref{eq:ratio_test}) reveals that the theoretical error envelope initially diverges for small $M$, initiating a strict monotonic decay towards zero only after the sequence ratio drops below $1$. Consequently, solving $U_{M+1}/U_M < 1$ yields the sufficient computational threshold for the onset of convergence:
\begin{equation}
	M > \frac{\pi L}{\lambda} - 0.5.
\end{equation}
 
\section{}
\label{appendix:2D_convergence}
To bound the first term in (\ref{eq:2D_error_triangle}), we define the marginal function $g(u) = \int_{-1}^{1} f(u, v)\,dv$. Applying the 1D GLQ error bound in (\ref{eq:E_N_f_4}) requires bounding the $2M_x$-th derivative of $g(u)$. By relaxing the integral to its global supremum over $v \in [-1, 1]$, we have:
\begin{align}
	\max_{u \in [-1,1]} |g^{(2M_x)}(u)| & \leq \int_{-1}^{1} \sup_{(u,v) \in \mathbb{R}^2} \left| \frac{\partial^{2M_x}f(u, v)}{\partial u^{2M_x}} \right|\,dv \nonumber \\
	& = 2 \sup_{(u,v)\in\mathbb{R}^2} \left| \frac{\partial^{2M_x}f(u, v)}{\partial u^{2M_x}} \right|.
\end{align}
Substituting this upper bound along with (\ref{eq:partial_u_bound}) into (\ref{eq:E_N_f_4}) yields:
\begin{equation}
	\left|E_{M_x}(g)\right| \leq 2 C(M_x) \left(\frac{e\pi L_x}{2M_x\lambda}\right)^{2M_x} \sup_{(u,v) \in \mathbb{R}^2}|f(u,v)|.
\end{equation}

For the second term in (\ref{eq:2D_error_triangle}), the 1D error $E_{M_y}(\cdot)$ evaluated along the specific $v$-axis at any quadrature node $u_m$ is initially bounded by the 1D sectional supremum $\sup_{v\in\mathbb{R}}|f(u_m, v)|$. To establish a unified global envelope independent of the nodes $\{u_m\}$, we relax this sectional bound to the global supremum over the entire 2D spatial domain, i.e., $\sup_{v\in\mathbb{R}} |f(u_m, v)| \leq \sup_{(u,v)\in\mathbb{R}^2} |f(u, v)|$. Since the Gauss-Legendre weights satisfy $\sum_{m=1}^{M_x} w_m^{(u)} = 2$, substituting this relaxed bound along with (\ref{eq:partial_v_bound}) yields:
\begin{align}
	&\sum_{m=1}^{M_x}w_m^{(u)}\left|E_{M_y}\big(f(u_m, \cdot)\big)\right| \leq \nonumber \\
	&\qquad\qquad\quad  2 C(M_y)\left(\frac{e\pi L_y}{2M_y\lambda}\right)^{2M_y} \sup_{(u,v)\in\mathbb{R}^2}|f(u,v)|.
\end{align}
Combining these two orthogonal bounds recovers the super-exponential upper bounds presented in the theorem, thereby concluding the proof.

\section{}
\label{appdix:Widom_conjecture}
We begin by introducing the multidimensional Szeg\H{o}-Widom asymptotic expansion~\cite{Widom1982}. Let $\Omega \subset \mathbb{R}^d$ and $\Gamma \subset \mathbb{R}^d$ be bounded, normalized spatial and wavenumber domains, respectively. To accommodate our 2D scenario (i.e., a rectangular spatial domain and a disk wavenumber domain), we assume the boundary of the spatial domain, $\partial\Omega$, is piecewise-smooth, while the boundary of the wavenumber domain, $\partial \Gamma$, is smooth. Note that for the 1D case (i.e., $d = 1$), $\partial\Omega$ and $\partial \Gamma$ degenerate into discrete point sets, to which the Szeg\H{o}-Widom asymptotic expansion also applies. Given the above settings, we have the following asymptotic expansion.

\noindent\textbf{Szeg\H{o}--Widom asymptotic expansion:}
Let $\Omega \subset \mathbb{R}^d$ and $\Gamma \subset \mathbb{R}^d$ be bounded, normalized spatial and wavenumber domains, respectively. Assume that the spatial boundary, $\partial\Omega$, is piecewise-smooth, and the wavenumber boundary, $\partial \Gamma$, is smooth. Let $\mathbf{n}_r$ and $\mathbf{n}_k$ denote the outward unit normal vectors on $\partial\Omega$ and $\partial \Gamma$, respectively. Consider the scaled spatial domain $\Omega_\alpha = \alpha\Omega$ with a scaling factor $\alpha > 0$. Let $\mathcal{T}_\alpha$ be the self-adjoint integral operator defined over the domains $\Omega_\alpha$ and $\Gamma$. Then for a given test function $f(t)$ defined on $[0, 1]$ satisfying $f(0) = 0$, the trace of the operator $f(\mathcal{T}_\alpha)$ takes the following asymptotic expansion as $\alpha \to \infty$:
\begin{align}
	&\mathrm{tr}(f(\mathcal{T}_\alpha)) = \alpha^d \frac{f(1)}{(2\pi)^d}\int_{\Omega}\int_{\Gamma}\,d\mathbf{k} \,d \mathbf{r} + \nonumber \\
	& \alpha^{d-1}\ln\alpha\frac{\mathcal{U}(f)}{(2\pi)^{d-1}}\oint_{\partial\Omega}\oint_{\partial \Gamma}|\mathbf{n}_r\cdot \mathbf{n}_k|dS_k dS_r + o(\alpha^{d-1} \ln \alpha),
	\label{eq:Widom_Conjecture}
\end{align}
where $\mathcal{U}(f)$ is a functional depending on the spectral distribution of the test function, given by:
\begin{equation}
	\mathcal{U}(f) = \frac{1}{4\pi^2}\int_{0}^{1}\frac{f(t) - tf(1)}{t(1-t)}\,dt.
	\label{eq:U_functional}
\end{equation}

For the piecewise-smooth rectangular domain with a discontinuous step function considered here, the rigorous validity of this expansion remains open. We therefore use it to construct the semi-analytical eDoF reference.
We first note that the integral operator $\mathcal{T}_\alpha$ is formulated as 
\begin{equation}
	\mathcal{T}_\alpha = P_{\alpha, \Omega}Q_\Gamma P_{\alpha, \Omega},
	\label{eq:Operator_T_alpha}
\end{equation}
where $P_{\alpha, \Omega}$ denotes the spatial projection operator defined as
\begin{equation}
	(P_{\alpha,\Omega} \psi)(\mathbf{r}) = \psi(\mathbf{r}) \text{ for } \mathbf{r} \in \alpha\Omega, \text{ and } 0 \text{ otherwise}.
	\label{eq:P_Omega_alpha}
\end{equation}
Note that $\mathbf{r}$ in (\ref{eq:P_Omega_alpha}) is defined in the physical spatial domain $\alpha\Omega$. However, in (\ref{eq:Widom_Conjecture}), to facilitate the asymptotic analysis as $\alpha \to \infty$, we adopt the normalized domain $\Omega$ via the coordinate transformation $\mathbf{r} \to \alpha\mathbf{r}$. The operator $Q_\Gamma$ represents the band-limiting operator, defined as
\begin{equation}
	(Q_\Gamma\psi)(\mathbf{r}) = \int_{\mathbb{R}^d}R(\mathbf{r} - \mathbf{r}')\psi(\mathbf{r}')\,d\mathbf{r}',
	\label{eq:Q_Gamma}
\end{equation}
where the kernel function $R(\mathbf{r})$ is the inverse Fourier transform of the indicator function over the wavenumber domain $\Gamma$:
\begin{equation}
	R(\mathbf{r}) = \frac{1}{(2\pi)^d}\int_{\Gamma}e^{j \mathbf{k}\cdot\mathbf{r}}\,d\mathbf{k}.
	\label{eq:kappa_r}
\end{equation}
Given (\ref{eq:Operator_T_alpha})-(\ref{eq:kappa_r}), the action of the self-adjoint operator $\mathcal{T}_{\alpha}$ on a function $\psi$ can be written as the following integral equation for $\mathbf{r} \in \alpha\Omega$:
\begin{equation}
	(\mathcal{T}_{\alpha}\psi)(\mathbf{r}) = \int_{\alpha\Omega}R(\mathbf{r}-\mathbf{r}')\psi(\mathbf{r}')\,d\mathbf{r}'.
\end{equation}

We emphasize that the evaluation of $f(1)$ in the first term of $\mathrm{tr}(f(\mathcal{T}_\alpha))$ is a simplification of $f(\sigma(\mathbf{k}, \mathbf{r}))$ from the general Widom formulation. Here $\sigma(\mathbf{k}, \mathbf{r})$ is the discontinuous symbol dependent on $\mathbf{k}$ and $\mathbf{r}$. In our considered scenario, $\sigma(\mathbf{k}, \mathbf{r})$ acts as an indicator function, where $\sigma(\mathbf{k}, \mathbf{r}) = 1$ when $(\mathbf{k}, \mathbf{r}) \in \Gamma \times \Omega$ and $0$ otherwise. 
To count the eigenvalues of the concentration operator $T_\alpha$ above a prescribed threshold $\epsilon \in (0, 1)$, we set the test function $f(t)$ as the step function:
\begin{equation}
	f(t) = 1 \text{ for } t \in [\epsilon, 1], \text{ and } 0 \text{ for } t \in [0, \epsilon).
	\label{eq:step_function}
\end{equation}
Since $f(\mathcal{T}_{\alpha})$ acts on the eigenvalues of the operator $\mathcal{T}_{\alpha}$, the operator trace yields exactly the corresponding eigenvalue count, i.e., the eDoF reference:
\begin{equation}
	\mathrm{tr}(f(\mathcal{T}_\alpha)) = \sum_{n=1}^{\infty}f(\mu_n) = N_{\text{eDoF}}(\epsilon).
	\label{eq:N_rank_epsilon}
\end{equation}

\section{}
\label{sec:1D_proof}
In this appendix, we consider the 1D case. The physical supports of the spatial and wavenumber domains are given by $[-L/2, L/2]$ and $[-\frac{2\pi}{\lambda}, \frac{2\pi}{\lambda}]$, respectively. We normalize both the spatial and wavenumber domains to $[-1, 1]$. Consequently, $\alpha = \frac{\pi L}{\lambda}$. Given the definitions of $f(t)$ and $\mathcal{U}(f)$ in (\ref{eq:step_function}) and (\ref{eq:U_functional}), respectively, we have $f(1) = 1$ and 
\begin{equation}
	\mathcal{U}(f) = \frac{1}{4\pi^2} \ln\left(\frac{1-\epsilon}{\epsilon}\right).
\end{equation}
Furthermore, since $\mathbf{n}_r = \pm 1$ and $\mathbf{n}_k = \pm 1$ in the 1D case, we have $|\mathbf{n}_r\cdot\mathbf{n}_k| = 1$. Recalling that the 1D boundaries are discrete point sets, the boundary integral in (\ref{eq:Widom_Conjecture}) simplifies to a discrete double summation over the boundary points:
\begin{equation}
	\oint_{\partial\Omega}\oint_{\partial \Gamma}|\mathbf{n}_r\cdot \mathbf{n}_k|d S_k d S_r = \sum_{r\in\partial\Omega}\sum_{k\in \partial \Gamma}|\mathbf{n}_r\cdot \mathbf{n}_k| = 4.
\end{equation}
Substituting the above evaluations into (\ref{eq:Widom_Conjecture}), we have
\begin{equation}
	N_{\text{eDoF}}^{(1D)}(\epsilon) = \frac{2L}{\lambda} + \frac{1}{\pi^2}\ln\Big(\frac{1-\epsilon}{\epsilon}\Big)\ln (c) + o\left(\ln (c) \right),
\end{equation}
where $c = \frac{\pi L}{\lambda}$ is the well-known space-bandwidth product~\cite{landau1980eigenvalue}.

\section{}
\label{sec:2D_proof}
In this appendix, we consider the 2D rectangular aperture case. 
Since $f(1) = 1$ is established by the discontinuous step function in (\ref{eq:step_function}), we can evaluate the leading-order double volume integral in (\ref{eq:Widom_Conjecture}) directly over the unscaled physical spatial and wavenumber domains, denoted by $\Omega_{\text{phys}}$ and $\Gamma_{\text{phys}}$, respectively. Specifically, we have
\begin{equation}
	\mathcal{I}_{\text{vol}} = \frac{1}{4\pi^2} \int_{\Omega_{\text{phys}}} \int_{\Gamma_{\text{phys}}} \,d\mathbf{k} \,d\mathbf{r} = \frac{\text{Area}(\Omega_{\text{phys}})\text{Area}(\Gamma_{\text{phys}})}{4\pi^2}.
	\label{eq:volume_term_SW}
\end{equation}
Substituting the physical area of the rectangular array $\text{Area}(\Omega_{\text{phys}}) = L_x L_y = A_{\text{array}}$ and the area of the isotropic wavenumber disk $\text{Area}(\Gamma_{\text{phys}}) = \pi \kappa^2 = \pi \left(\frac{2\pi}{\lambda}\right)^2 = \frac{4\pi^3}{\lambda^2}$ yields:
\begin{equation}
	\mathcal{I}_{\text{vol}} = \frac{A_{\text{array}}\Big( \frac{4\pi^3}{\lambda^2} \Big)}{4\pi^2} = \frac{\pi A_{\text{array}}}{\lambda^2}.
	\label{eq:2D_leading_term}
\end{equation}

Anisotropic Boundary Decoupling: For a planar rectangular array, the space-bandwidth products ($c_x = \pi\frac{L_x}{\lambda}$ and $c_y = \pi\frac{L_y}{\lambda}$) differ along orthogonal axes. 
Consequently, we decompose the boundary correction term $\mathcal{I}_{\text{bound}}$ into two contributions associated with the vertical and horizontal edges, denoted as $\partial\Omega_y$ and $\partial\Omega_x$, respectively:
\begin{equation}
	\mathcal{I}_{\text{bound}} = \text{Edge}_y + \text{Edge}_x.
\end{equation}
For the two vertical edges $\partial\Omega_y$ (total length $2L_y$), the normal vector is $\mathbf{n}_r = (\pm 1, 0)$. Since the normal vector for the circular wavenumber domain is $\mathbf{n}_k = (\cos\theta_k, \sin\theta_k)$, the wavenumber domain boundary integral yields:
\begin{equation}
	\oint_{\partial \Gamma} |\mathbf{n}_r \cdot \mathbf{n}_k| \,dS_k = \int_0^{2\pi} |\cos\theta_k| \kappa \,d\theta_k = 4\kappa.
	\label{eq:vertical_oint}
\end{equation}
Because the spatial truncation associated with these vertical edges occurs along the $x$-direction, we associate the corresponding logarithmic scaling factor with the space-bandwidth product along the $x$-axis, i.e., $\ln(c_x) = \ln\left(\pi \frac{L_x}{\lambda}\right)$.
Substituting the closed-form expression of $\mathcal{U}(f)$ and integrating over the vertical boundaries $\partial\Omega_y$, the logarithmic correction $\text{Edge}_y$ is given by:
\begin{align}
	\text{Edge}_y = & \ln(c_x)\frac{\mathcal{U}(f)}{2\pi}\int_{\partial\Omega_y}\oint _{\partial \Gamma}|\mathbf{n}_r \cdot \mathbf{n}_k| \,dS_k\,d l_y \nonumber \\
	= & \frac{2 L_y}{\pi^2 \lambda} \ln\Big(\pi \frac{L_x}{\lambda}\Big) \ln\Big(\frac{1-\epsilon}{\epsilon}\Big).
	\label{eq:Edge_y}
\end{align}
By symmetry, for the two horizontal edges $\partial\Omega_x$ (total length $2L_x$ with $\mathbf{n}_r = (0, \pm 1)$), the spatial truncation occurs along the $y$-direction, invoking the logarithmic scaling factor $\ln(c_y) = \ln\big(\pi \frac{L_y}{\lambda}\big)$. The analogous evaluation yields the boundary contribution $\text{Edge}_x$:
\begin{equation}
	\text{Edge}_x = \frac{2 L_x}{\pi^2 \lambda} \ln\Big(\pi \frac{L_y}{\lambda}\Big) \ln\Big(\frac{1-\epsilon}{\epsilon}\Big).
	\label{eq:Edge_x}
\end{equation}
Combining the volume term in \eqref{eq:2D_leading_term} with the two boundary corrections in \eqref{eq:Edge_y} and \eqref{eq:Edge_x} yields the semi-analytical eDoF reference in \eqref{eq:N_rank_2D}.

\section{}
\label{appendix:eigenvalue_ordering}
For any $\psi \in L^2(\Omega_{\text{phys}})$, using the kernel representations of $K_{\mathrm{con}}$ and $K_{\mathrm{iso}}$ in \eqref{eq:K_con} and \eqref{eq:K_iso}, respectively, we obtain
\begin{align}
	&\int_{\Omega_{\text{phys}}}\psi^*(\mathbf{r})
	\left[\left(
	\mathcal{T}_{\mathrm{iso}} - 
	\frac{\lambda^2}{2\pi}\mathcal{T}_{\mathrm{con}}
	\right)
	\psi
	\right](\mathbf{r})
	\, d\mathbf{r} \nonumber \\
	&= 
	\iint_{k_x^2 + k_y^2 \leq \kappa^2}
	\left[
	\frac{1}{2\pi\kappa\sqrt{\kappa^2 - k_x^2 - k_y^2}}
	-\frac{1}{2\pi\kappa^2}
	\right] \nonumber \\
	& \times 
	\int_{\Omega_{\text{phys}}}
	\int_{\Omega_{\text{phys}}}
	\psi^*(\mathbf{r})
	e^{j(
		k_x(x-x') + k_y(y-y')
		)}
	\psi(\mathbf{r}')
	\, d\mathbf{r}'
	\, d\mathbf{r}
	\, dk_x \, dk_y \nonumber \\
	&= 
	\iint_{k_x^2 + k_y^2 \leq \kappa^2}
	\left[
	\frac{1}{2\pi\kappa\sqrt{\kappa^2 - k_x^2 - k_y^2}}
	-\frac{1}{2\pi\kappa^2}
	\right] \nonumber \\
	&\qquad \times \left|
	\int_{\Omega_{\text{phys}}}\psi(\mathbf{r})
	e^{-j(k_x x + k_y y)} 
	d\mathbf{r}
	\right|^2
	\, dk_x \, dk_y \geq 0,
\end{align}
where the final inequality follows from \(\sqrt{\kappa^2-k_x^2-k_y^2}\leq\kappa\) almost everywhere over the propagating wavenumber disk. Therefore, we have $\mathcal{T}_{\mathrm{iso}} \succeq \frac{\lambda^2}{2\pi}\mathcal{T}_{\mathrm{con}}$. Since \(K_{\rm con}\) and \(K_{\rm iso}\) are continuous Hermitian kernels over the bounded aperture, the corresponding operators are compact and self-adjoint. The max--min principle~\cite[Thm.~4.22]{borthwick2020spectral} gives 
\begin{equation}
	\mu_n^{\mathrm{iso}} 
	\geq
	\frac{\lambda^2}{2\pi}
	\mu_n^{\mathrm{con}},
	\quad
	n = 1, 2, \ldots
\end{equation}
By normalizing the eigenvalues of $\mathcal{T}_{\mathrm{con}}$ and $\mathcal{T}_{\mathrm{iso}}$ with respect to $\mu_1^{\mathrm{con}}$ and $\mu_1^{\mathrm{iso}}$, we obtain the result in \eqref{eq:one-side_ordering}. This completes the proof.

\bibliographystyle{IEEEtran}
\bibliography{mybibliography.bib}

\begin{thebibliography}{10}
\providecommand{\url}[1]{#1}
\csname url@samestyle\endcsname
\providecommand{\newblock}{\relax}
\providecommand{\bibinfo}[2]{#2}
\providecommand{\BIBentrySTDinterwordspacing}{\spaceskip=0pt\relax}
\providecommand{\BIBentryALTinterwordstretchfactor}{4}
\providecommand{\BIBentryALTinterwordspacing}{\spaceskip=\fontdimen2\font plus
\BIBentryALTinterwordstretchfactor\fontdimen3\font minus
  \fontdimen4\font\relax}
\providecommand{\BIBforeignlanguage}[2]{{%
\expandafter\ifx\csname l@#1\endcsname\relax
\typeout{** WARNING: IEEEtran.bst: No hyphenation pattern has been}%
\typeout{** loaded for the language `#1'. Using the pattern for}%
\typeout{** the default language instead.}%
\else
\language=\csname l@#1\endcsname
\fi
#2}}
\providecommand{\BIBdecl}{\relax}
\BIBdecl

\bibitem{Wei2026Electromagnetic}
L.~Wei \emph{et~al.}, ``{Electromagnetic Information Theory for Holographic
  MIMO Communications},'' \emph{IEEE Communications Surveys \& Tutorials},
  vol.~28, pp. 6211--6240, 2026.

\bibitem{Zhang2026Wavenumber}
Z.~Zhang \emph{et~al.}, ``{Wavenumber-Domain Signal Processing for Holographic
  MIMO: Foundations, Methods, and Future Directions},'' \emph{IEEE
  Communications Standards Magazine}, vol.~10, no.~2, pp. 127--133, 2026.

\bibitem{Bjornson2024Towards}
E.~Bjornson \emph{et~al.}, ``{Towards 6G MIMO: Massive Spatial Multiplexing,
  Dense Arrays, and Interplay Between Electromagnetics and Processing},''
  \emph{arXiv:2401.02844}, 2024.

\bibitem{Hu2018Beyond}
S.~Hu, F.~Rusek, and O.~Edfors, ``{Beyond Massive MIMO: The Potential of Data
  Transmission With Large Intelligent Surfaces},'' \emph{IEEE Transactions on
  Signal Processing}, vol.~66, no.~10, pp. 2746--2758, 2018.

\bibitem{Dardari2020Commu}
D.~Dardari, ``Communicating with large intelligent surfaces: Fundamental limits
  and models,'' \emph{IEEE Journal on Selected Areas in Communications},
  vol.~38, no.~11, pp. 2526--2537, 2020.

\bibitem{Gong2024HMIMO}
T.~Gong \emph{et~al.}, ``{Holographic MIMO Communications: Theoretical
  Foundations, Enabling Technologies, and Future Directions},'' \emph{IEEE
  Communications Surveys \& Tutorials}, vol.~26, no.~1, pp. 196--257, 2024.

\bibitem{Liu2025CAPA}
Y.~Liu, C.~Ouyang, Z.~Wang, J.~Xu, X.~Mu, and Z.~Ding, ``{CAPA:
  Continuous-Aperture Arrays for Revolutionizing 6G Wireless Communications},''
  \emph{IEEE Wireless Communications}, vol.~32, no.~4, pp. 38--45, 2025.

\bibitem{Pizzo2020Spatially}
A.~Pizzo \emph{et~al.}, ``{Spatially-Stationary Model for Holographic MIMO
  Small-Scale Fading},'' \emph{IEEE Journal on Selected Areas in
  Communications}, vol.~38, no.~9, pp. 1964--1979, 2020.

\bibitem{Bucci2025Degrees}
O.~M. Bucci and M.~D. Migliore, ``Degrees of freedom and sampling
  representation of electromagnetic fields: Concepts and applications,''
  \emph{IEEE Antennas and Propagation Magazine}, vol.~67, no.~3, pp. 10--22,
  2025.

\bibitem{Pizzo2022Nyquist}
A.~Pizzo, A.~d.~J. Torres, L.~Sanguinetti, and T.~L. Marzetta, ``{Nyquist
  Sampling and Degrees of Freedom of Electromagnetic Fields},'' \emph{IEEE
  Transactions on Signal Processing}, vol.~70, pp. 3935--3947, 2022.

\bibitem{Demir2022Channel}
{\"O}.~T. Demir, E.~Bj{\"o}rnson, and L.~Sanguinetti, ``{Channel Modeling and
  Channel Estimation for Holographic Massive MIMO With Planar Arrays},''
  \emph{IEEE Wireless Communications Letters}, vol.~11, no.~5, pp. 997--1001,
  2022.

\bibitem{Wan2023Can}
Z.~Wan, J.~Zhu, and L.~Dai, ``{Can Continuous Aperture MIMO Obtain More Mutual
  Information Than Discrete MIMO?}'' \emph{IEEE Communications Letters},
  vol.~27, no.~12, pp. 3185--3189, 2023.

\bibitem{Pizzo2022Fourier}
A.~Pizzo \emph{et~al.}, ``{Fourier Plane-Wave Series Expansion for Holographic
  MIMO Communications},'' \emph{IEEE Transactions on Wireless Communications},
  vol.~21, no.~9, pp. 6890--6905, 2022.

\bibitem{Yan2026Spectrally}
H.~Yan \emph{et~al.}, ``{A Spectrally Convergent Discretization of Holographic
  MIMO Channels},'' \emph{IEEE Communications Letters}, vol.~30, pp.
  1964--1968, 2026.

\bibitem{landau1980eigenvalue}
H.~J. Landau \emph{et~al.}, ``Eigenvalue distribution of time and frequency
  limiting,'' \emph{Journal of Mathematical Analysis and Applications},
  vol.~77, no.~2, pp. 469--481, 1980.

\bibitem{Widom1982}
H.~Widom, \emph{On a Class of Integral Operators with Discontinuous
  Symbol}.\hskip 1em plus 0.5em minus 0.4em\relax Basel: Birkh{\"a}user Basel,
  1982, pp. 477--500.

\bibitem{PhysRevLett.96.100503}
D.~Gioev \emph{et~al.}, ``{Entanglement Entropy of Fermions in Any Dimension
  and the Widom Conjecture},'' \emph{Phys. Rev. Lett.}, vol.~96, p. 100503, Mar
  2006.

\bibitem{sobolev2013pseudo}
A.~V. Sobolev, \emph{Pseudo-differential operators with discontinuous symbols:
  Widom’s conjecture}.\hskip 1em plus 0.5em minus 0.4em\relax Memoirs of the
  American Mathematical Society, 2013, vol. 222, no. 1043.

\bibitem{Sobolev2015Wiener}
------, ``{Wiener–Hopf Operators in Higher Dimensions: The Widom Conjecture
  for Piece-Wise Smooth Domains},'' \emph{Integral Equations and Operator
  Theory}, vol.~81, no.~3, pp. 435--449, Mar 2015.

\bibitem{Sobolev2017Functions}
------, ``Functions of self-adjoint operators in ideals of compact operators,''
  \emph{Journal of the London Mathematical Society}, vol.~95, no.~1, pp.
  157--176, 2017.

\bibitem{hildebrand1962advanced}
F.~Hildebrand, \emph{{Advanced Calculus for Applications}}.\hskip 1em plus
  0.5em minus 0.4em\relax Upper Saddle River, NJ, USA, Prentice-Hall, 1962.

\bibitem{riesz1990functional}
F.~Riesz \emph{et~al.}, \emph{Functional Analysis}.\hskip 1em plus 0.5em minus
  0.4em\relax New York: Courier Corporation, 1990.

\bibitem{ghojogh2021reproducing}
B.~Ghojogh \emph{et~al.}, ``{Reproducing Kernel Hilbert Space, Mercer's
  Theorem, Eigenfunctions, Nystrom Method, and Use of Kernels in Machine
  Learning: Tutorial and Survey},'' \emph{arXiv preprint arXiv:2106.08443},
  2021.

\bibitem{Slepian1961PSWF}
D.~Slepian \emph{et~al.}, ``{Prolate spheroidal wave functions, Fourier
  analysis and uncertainty — I},'' \emph{The Bell System Technical Journal},
  vol.~40, no.~1, pp. 43--63, 1961.

\bibitem{nikolskii1975approximation}
S.~M. Nikol'skii, \emph{Approximation of Functions of Several Variables and
  Imbedding Theorems}.\hskip 1em plus 0.5em minus 0.4em\relax New York:
  Springer-Verlag, 1975.

\bibitem{Rahman2009Bernstein}
Q.~Rahman \emph{et~al.}, ``{On Bernstein's inequality for entire functions of
  exponential type},'' \emph{Journal of Mathematical Analysis and
  Applications}, vol. 359, no.~1, pp. 168--180, 2009.

\bibitem{lanczos1950iteration}
C.~Lanczos, ``{An iteration method for the solution of the eigenvalue problem
  of linear differential and integral operators},'' \emph{Journal of Research
  of the National Bureau of Standards}, vol.~45, no.~4, pp. 255--282, 1950.

\bibitem{Pizzo2022Spatial}
A.~Pizzo \emph{et~al.}, ``{Spatial Characterization of Electromagnetic Random
  Channels},'' \emph{IEEE Open Journal of the Communications Society}, vol.~3,
  pp. 847--866, 2022.

\bibitem{Wang2022Electromagnetic}
T.~Wang \emph{et~al.}, ``{Electromagnetic-Compliant Channel Modeling and
  Performance Evaluation for Holographic MIMO},'' in \emph{2022 IEEE Globecom
  Workshops (GC Wkshps)}, 2022, pp. 747--752.

\bibitem{3GPP38901}
{3GPP}, ``{Study on Channel Model for Frequencies from 0.5 to 100 GHz},'' 3rd
  Generation Partnership Project (3GPP), Technical Report TR 38.901, 2026,
  version 19.4.0, Release 19.

\bibitem{Tubic2024Correlation}
K.~Turbic, M.~Kasparick, and S.~Stańczak, ``{Correlation Properties in
  Channels With von Mises-Fisher Distribution of Scatterers},'' \emph{IEEE
  Wireless Communications Letters}, vol.~13, no.~12, pp. 3638--3642, 2024.

\bibitem{davis2007methods}
P.~J. Davis \emph{et~al.}, \emph{Methods of numerical integration}.\hskip 1em
  plus 0.5em minus 0.4em\relax Courier Corporation, 2007.

\bibitem{hildebrand1987introduction}
F.~B. Hildebrand, \emph{Introduction to Numerical Analysis}, 2nd~ed.\hskip 1em
  plus 0.5em minus 0.4em\relax New York: Dover Publications, 1987.

\bibitem{robbins1955remark}
H.~Robbins, ``{A remark on Stirling's formula},'' \emph{The American
  Mathematical Monthly}, vol.~62, no.~1, pp. 26--29, 1955.

\bibitem{rudin1974real}
W.~Rudin, \emph{Real and Complex Analysis}.\hskip 1em plus 0.5em minus
  0.4em\relax McGraw-Hill, 1974.

\bibitem{borthwick2020spectral}
D.~Borthwick, \emph{Spectral Theory: Basic Concepts and Applications, Graduate
  Texts in Mathematics}.\hskip 1em plus 0.5em minus 0.4em\relax Springer, 2020,
  vol. 284.

\end{thebibliography}

\end{document}